\documentclass{article}
\usepackage{graphicx}
\usepackage{amsfonts, amsmath, amsthm, amssymb}
\usepackage{booktabs}
\usepackage{xcolor}
\usepackage{mathtools}
\usepackage{setspace} 
\usepackage[left=1in,top=1in,right=1in,bottom=1in]{geometry}
\usepackage{multirow}
\usepackage{algorithm}
\usepackage{algorithmic}
\usepackage[round]{natbib}
\usepackage{longtable}
\usepackage{array}
\usepackage[colorlinks=true, linkcolor=blue, citecolor=blue, urlcolor=blue]{hyperref}
\usepackage{pifont} 
\makeatletter
\newenvironment{breakablealgorithm}
  {%
   \begin{center}
     \refstepcounter{algorithm}%
     \hrule height.8pt depth0pt \kern2pt%
     \renewcommand{\caption}[2][\relax]{%
       {\raggedright\textbf{Algorithm \thealgorithm} ##2\par}%
       \ifx\relax##1\relax
         \addcontentsline{loa}{algorithm}{\protect\numberline{\thealgorithm}##2}%
       \else
         \addcontentsline{loa}{algorithm}{\protect\numberline{\thealgorithm}##1}%
       \fi
       \kern2pt\hrule\kern2pt
     }%
  }{%
     \kern2pt\hrule\relax
   \end{center}
  }
\makeatother

\newtheorem{remark}{Remark}
\newtheorem{theorem}{Theorem}[section]
\newtheorem{corollary}{Corollary}[theorem]
\newtheorem{lemma}[theorem]{Lemma}

\newcommand{\eval}[0]{\biggr\rvert}
\newcommand{\RR}[0]{\mathbb{R}}

\newcommand{\norm}[1]{\left\lVert#1\right\rVert}

\title{Improving the Efficiency of Subgroup Analysis in Randomized Controlled Trials with TMLE}
\author{Sky Qiu$^{0,1,2}$, Nerissa Nance$^{2,3}$, Rachael Phillips$^2$, Jens Tarp$^{3}$, \\Maya Petersen$^{1,2}$, and Mark van der Laan$^{1,2}$}
\date{\today}

\begin{document}
\maketitle
\begin{abstract}

Subgroup analyses within randomized controlled trials are often underpowered due to limited sample sizes. We address this challenge by leveraging trial participants outside the subgroup of interest to augment estimation within the subgroup. Specifically, we study two Targeted Maximum Likelihood Estimators (TMLEs) that borrow information from non-subgroup participants within the same trial: a TMLE with pooled regression (TMLE-PR) and an Adaptive Targeted Maximum Likelihood Estimator (A-TMLE). Both estimators enable information sharing without relying on any external real-world data, thereby capitalizing on key strengths of the trial: most importantly, the protection against bias afforded by the randomized treatment, but also harmonized data collection, and consistent treatment and outcome definitions. The general strategy proposed here directly advances the priorities of key regulatory agencies, including the FDA, by improving the precision of subgroup-specific treatment effect estimates without introducing external sources of bias, thereby facilitating rigorous inference to support equitable labeling, access, and post-market evaluation. Our simulation study suggests that both TMLE-PR and A-TMLE improve efficiency for estimation of subgroup effects in finite samples, compared to both conventional unadjusted estimator and semi-parametrically efficient subgroup-only estimators. In a case study based on analysis of data from a cardiovascular outcome trial (LEADER, NCT01179048), we estimate the risk reduction of major adverse cardiac events (MACE) under liraglutide treatment among Black and Asian subgroups---each comprising less than 10\% of the trial population---using the proposed estimators that borrow information from the remainder of the trial. Using A-TMLE, in particular, we find estimated absolute MACE risk reductions of 1.6, 1.5, and 1.5 percentage points among Asian participants and 2.1, 2.0, and 2.1 percentage points among Black participants at 365, 540, and 730 days, respectively, with 95\% confidence intervals excluding the null at each time point. These results highlight TMLE-PR and A-TMLE as practical approaches for enhancing the efficiency and precision of subgroup-specific evidence generation using only data collected within the same trial.
\end{abstract}

\section{Introduction}
\footnotetext[0]{Correspondence: \tt sky.qiu@berkeley.edu}
\footnotetext[1]{Division of Biostatistics, School of Public Health, University of California, Berkeley, CA, USA}
\footnotetext[2]{Center for Targeted Machine Learning and Causal Inference, School of Public Health, University of California, Berkeley, CA, USA}
\footnotetext[3]{Novo Nordisk A/S, Bagsvaerd, Denmark}

Estimating treatment effects within specific patient subgroups in randomized controlled trials (RCTs) is important for characterizing subgroup-specific average treatment effects, which may differ from the population average treatment effect \citep{kravitz_evidence_2004,kent_assessing_2010}. However, RCTs often lack sufficient representation of certain subpopulations of interest, limiting the statistical power of subgroup-specific analyses \citep{wang_statistical_2021}. Real-world data (RWD) sources may contain a large number of patients in the subgroup of interest and therefore offer a potential avenue to augment trial subpopulations and enhance inference for subgroup effects. However, integrating RWD with RCT data introduces substantial statistical challenges, particularly in controlling bias arising from incompatibilities between data sources. These challenges include distributional shifts and structural differences in how data are generated and recorded, such as variations in patient populations and eligibility criteria, treatment implementation and adherence, endpoint definitions, confounding mechanisms, and coding systems \citep{fda_rwe_2023}. Therefore, it is of interest to explore methods for improving the precision of subgroup analyses using trial data alone, without reliance on external RWD.

When restricting attention to RCT data alone, one approach to conduct subgroup analysis is to consider whether information from patients outside the subgroup, yet still enrolled in the trial, can be leveraged to improve estimation of the subgroup-specific treatment effect. The idea of improving subgroup-specific treatment effect estimates by borrowing information from the broader trial population has a long history in Bayesian and empirical-Bayes subgroup analysis. For example, \cite{davis_empirical_1990} proposed empirical-Bayes estimators that ``shrink the extreme estimates toward the overall measure of treatment difference,'' while \cite{dixon_bayesian_1992} formulated Bayesian subset-analysis models that shrink treatment-by-subgroup interaction effects toward zero, thereby shrinking subgroup-specific effects toward a common treatment effect. \cite{simon_bayesian_2002} later described this idea as estimating subset-specific effects as an average of observed within-subset differences and overall differences. 

In this article, we introduce two new estimators for subgroup analysis. Our proposals are motivated by the same information-borrowing principle, but target a prespecified subgroup causal estimand through a semi-parametric analogue of classical Bayesian/empirical-Bayes partial pooling for subgroup effects \citep{jones_bayesian_2011,henderson_bayesian_2016,wang_bayesian_2024,wolbers_using_2025,wolbers_unified_2026}. Our proposed estimators fall under the Targeted Maximum Likelihood Estimation (TMLE) framework, which is a general statistical estimation framework for constructing asymptotically linear, efficient, doubly- or multiply-robust, plug-in estimators \citep{vdl_targeted_2006,vdl_targeted_2011,vdl_targeted_2018}. Given sample size limitations, efficient estimators are essential in subgroup analysis. TMLE allows for the integration of flexible machine learning algorithms, through, for example, Super Learner \citep{sl_vdl_2007}, into its nuisance parameter estimation, thereby avoiding making any restrictive parametric modeling assumptions.

Our first proposed estimator works by fitting an initial regression of the outcome on the trial arm and baseline covariates on the pooled trial sample, evaluating that regression on the subgroup sample, and then targeting the subgroup-specific estimand with a standard TMLE. We refer to this estimator as TMLE with pooled regression (TMLE-PR). TMLE-PR may improve upon a subgroup-only TMLE by obtaining a better estimate of the subgroup-specific outcome regression, which is a key component that drives the variance of the estimator.

Our second proposed estimator is inspired by data integration methods \citep{shi_data_2023,colnet_causal_2024,lin_data_2024}, originally developed and often applied in hybrid studies that augment an RCT with RWD to improve power. One such approach is Adaptive TMLE (A-TMLE) \citep{vanderlaan2024atmle}. Note that A-TMLE is a general estimation framework that aims to regularize the TMLE targeting step in challenging settings with, e.g., limited overlap in covariate distribution between the treatment and control arms \citep{van2023adaptive}. The version of A-TMLE considered here is a specific one tailored to the data integration setting \citep{vanderlaan2024atmle}, with the goal of attaining efficiency gains in finite samples. In the data integration setting, this A-TMLE decomposes the target estimand into (1) a pooled estimand, which combines information from the target and external sources ignoring potential biases, and (2) a bias estimand, which captures the bias induced by incorporating external data. This A-TMLE can then be viewed as performing a bias correction to the pooled estimand. In the subgroup analysis context, we may conceptualize the subgroup as the target source and patients outside the subgroup (but still in the RCT) as an external source, and apply data integration estimators to improve the efficiency of treatment effect estimation within the subgroup. This way, we draw on the internal infrastructure and environment of the RCT, where data quality, randomization, and measurement protocols are already harmonized, while correcting for subgroup differences that are not captured by measured baseline covariates. By regarding the rest of the trial population as ``external,'' we sidestep many of the comparability issues that arise when integrating RCT data with external RWD, while possibly still gaining efficiency in finite samples through statistical borrowing. The second estimator we propose is a version of this A-TMLE tailored to subgroup analysis. For subgroup analysis, the pooled estimand is naturally defined over the full RCT population, while the bias estimand captures the difference between the subgroup-specific target estimand and the pooled estimand. Such an A-TMLE can be used to borrow information from the non-subgroup individuals to improve the estimation of subgroup-specific treatment effects. We emphasize, however, that conducting subgroup analysis in the manner described in this article is estimator-agnostic. In principle, any data integration estimator developed that augments an RCT with RWD to improve efficiency of treatment effect estimations in trials could potentially be adapted to conduct subgroup analysis in this fashion. 

In simulation studies, TMLE-PR and A-TMLE show improved efficiency (as measured in mean-squared-error) compared to both unadjusted and covariate-adjusted subgroup-only estimators. In particular, A-TMLE tends to offer larger gain in efficiency compared to TMLE-PR in settings where baseline covariates are strongly predictive of subgroup membership. We further illustrate the utility and practicality of this strategy using the Liraglutide Effect and Action in Diabetes (LEADER) trial, a phase 3B international, randomized, double-blind, placebo-controlled cardiovascular safety trial sponsored by Novo Nordisk \citep{marso2013design,marso_liraglutide_2016}. The trial enrolled patients with type 2 diabetes mellitus (T2DM) and followed them for up to five years, measuring a composite primary endpoint of first evidence of major adverse cardiac event (MACE). While the trial enrolled 9,340 patients and was adequately powered for its primary safety endpoint, subgroup analyses were generally underpowered (as is common in such settings). No significant subgroup effects by relevant demographic characteristics, such as race, were found in the main trial analysis. While there is no expected biological effect of race, a recent exploratory meta-analysis of relevant glucagon-like peptide 1 receptor agonist (GLP1) cardiovascular trials, including LEADER, found a stronger effect among Asian race \citep{lee_comparative_2025}. Several mechanisms have been hypothesized, including greater glycemic response to incretin-based therapy in Asian populations, type 2 diabetes pathophysiology characterized by greater beta-cell dysfunction and metabolic risk at lower BMI, differences in diet and postprandial glucose excursions, and potentially larger cardiometabolic consequences of a given amount of weight loss \citep{kim_hba1c_2014,kang_asian_2019,davis_race_2022}. At the same time, evidence for these mechanisms is mixed; for example, an EXSCEL analysis did not find clear race-specific differences in short-term glycemic, blood-pressure, or lipid responses to once-weekly exenatide, apart from a larger pulse-rate response among Asian participants \citep{davis_race_2022}. While pooling different trials can increase power, it comes with expected drawbacks of treatment, population, and trial environment heterogeneity, which may pose the risk of biasing the subgroup-specific treatment effects. To address the question of differential race subgroup effects while maintaining the integrity of the original RCT design, we apply our proposed estimator to estimate subgroup-specific treatment effects within the LEADER trial, drawing strength from the remainder of the trial population while adjusting for subgroup-specific differences. 

This article is organized as follows. We begin in Section \ref{sec:setup} with the problem setup, where we introduce the observed data structure, define the subgroup-specific causal estimand, and state the identification assumptions. In Section \ref{sec:tmle_pr}, we present TMLE-PR. Next, in Section \ref{sec:atmle}, we develop the subgroup A-TMLE. Since A-TMLE is a relatively more sophisticated estimator to implement, we share some practical implementation recommendations in Section \ref{sec:atmle_practical}. We then report simulation studies in Section \ref{sec:simulation} comparing unadjusted and covariate-adjusted subgroup-only estimators with TMLE-PR and A-TMLE. Finally, in Section \ref{sec:LEADER}, we apply the proposed methods to subgroup analyses in the LEADER trial and conclude in Section \ref{sec:discussion} with a discussion.

\section{Problem setup}\label{sec:setup}

We begin by introducing notation for the observed data. We then define the causal estimand of interest, which is the risk difference of a study outcome---in the LEADER trial, MACE---between the treatment and control arm at a user-specified time $\tau$ among the subgroup of interest. Next, we present a set of identification assumptions under which this causal estimand can be expressed as a functional of the observed data distribution $P_0$, thereby yielding a statistical estimand. 

\subsection{Observed data}

Let $T$ denote the time to first occurrence of MACE, $C$ denote the censoring time, and let $\Delta\in\{0,1\}$ be the event indicator, where $\Delta=1$ indicates that the event is observed and $\Delta=0$ indicates right-censoring. The observed data for each individual can thus be written as $O=(S,W,A,\tilde{T}=\min(T,C),\Delta)$, where $S\in\mathcal{S}=\{0,1\}$ indicates subgroup membership, $W\in\mathcal{W}\subset\RR^d$ is a vector of baseline covariates, and $A\in\mathcal{A}=\{0,1\}$ is a binary treatment. In the LEADER trial, participants were closely monitored with minimal loss to follow-up \citep{marso_liraglutide_2016}. Therefore, we make the simplifying assumption of no loss to follow-up. Our outcome of interest is a binary indicator of MACE occurrence by a fixed time horizon $\tau$, namely, $Y=I(T\leq\tau)$, and for the remainder of the paper, we use this simplified observed data notation $O=(S,W,A,Y)$. We adopt the notation $Pf\equiv\int f(O)dP$ from empirical process literature. For example, $P_nf\equiv 1/n\sum_{i=1}^nf(O_i)$.

Our observed data consists of $n$ copies of $O$: $O_1,\dots,O_n$ drawn i.i.d. from a probability distribution $P_0$ that lies in a statistical model $\mathcal{M}$. For our statistical model, we have knowledge that the treatment in the trial is 1:1 randomized, i.e., $P_0(A=1\mid W)=P_0(A=1)=1/2$, and that the subgroup membership is independent of treatment, i.e., $S\perp A$. We assume the following Structural Causal Model (SCM) \citep{pearl_causality_2009}:
\begin{align*}
S&=f_S(U_S),\\
W&=f_W(S,U_W),\\
A&=f_A(U_A),\\
Y&=f_Y(W,S,A,U_Y),
\end{align*}
where $f_S$, $f_W$, $f_Y$ are unrestricted unknown functions and $(U_S,U_W,U_A,U_Y)$ is a vector of exogenous errors. The function $f_A$ is known and corresponds to a randomized assignment mechanism, that is, e.g., $A=I(U_A\leq 1/2)$ with $U_A\sim\text{Uniform}(0,1)$, implying that $A$ is independent of $(W,S)$. This SCM encodes knowledge that subgroup membership $S$ may influence baseline covariates $W$ and outcome $Y$; and baseline covariates $W$ and treatment $A$ may causally affect the outcome $Y$. Potential outcomes are defined via interventions on the treatment node $A$. Specifically, for $a\in\mathcal{A}$, we define $Y_1=f_Y(W,S,1,U_Y)$ and $Y_0=f_Y(W,S,0,U_Y)$, the outcomes that would have been observed for a randomly sampled individual had they, possibly contrary to fact, been assigned to treatment or control, respectively \citep{pearl_causality_2009}.

Let $P^F$ denote the joint distribution of $(S,W,A,Y_0,Y_1)$, and let $\mathcal M^F$ denote the corresponding full-data model. Under consistency, the observed-data distribution $P$ is the distribution of $O=(S,W,A,Y_A)$ induced by $P^F$. We reserve $\mathcal M$ for the statistical model of these observed-data distributions. Write $P_0^F$ for the true full-data distribution and $P_0$ for its corresponding observed-data distribution.

\begin{remark}
We note that the notation used here mirrors \cite{vanderlaan2024atmle}, apart from $S$ denoting subgroup membership here instead of RCT membership in the original article.
\end{remark}

\subsection{Causal estimand and its identification}

Our target causal estimand $\Psi^F:\mathcal M^F\to\RR$ is the subgroup-specific risk difference of MACE between the treatment and control arms at time $\tau$. For $P^F\in\mathcal M^F$ with $P^F(S=1)>0$, define
$$
\Psi^F(P^F):=E_{P^F}(Y_1-Y_0\mid S=1).
$$
By the law of iterated expectations,
$$
\Psi^F(P^F)=E_{P^F}\left[E_{P^F}(Y_1-Y_0\mid S=1,W)\mid S=1\right].
$$
This equality requires no causal identification assumptions. We retain the iterated-expectation representation to emphasize that the conditional treatment effect is averaged over the covariate distribution within the subgroup, rather than over the pooled covariate distribution.

\begin{remark}\label{rem:avg_over_pooled}
In the original A-TMLE framework for data integration \citep{vanderlaan2024atmle}, the target causal estimand is defined as
$$
E_{P_W^F}\left[E_{P^F}(Y_1-Y_0\mid S=1,W)\right],
$$
which differs from our definition here in that the outer expectation is taken over the marginal baseline covariate distribution $P_W^F$, corresponding to the pooled population (that is, the RCT enrolled population in the current context). In contrast, our goal here is to characterize the treatment effect within the subgroup itself. Accordingly, we define the target estimand with respect to $P_{W\mid S=1}^F$, which is more appropriate for subgroup-specific inference. Notably, the causal estimand $\Psi^F$ introduced here was also mentioned in \cite{vanderlaan2024atmle}, and it is denoted by $\Psi_2$.
\end{remark}

The following identification assumptions are needed to nonparametrically identify the causal estimand:
\begin{itemize}
\item[\textbf{A1}] \textbf{No unmeasured confounding in the subgroup.} Under $P_0^F$, $A\perp(Y_0,Y_1)\mid W,S=1$;
\item[\textbf{A2}] \textbf{Positivity of treatment assignment in the subgroup.} $0<P_0(A=a\mid S=1,W=w)<1$ for $a\in\mathcal A$ and $P_{0,W\mid S=1}$-almost every $w$.
\end{itemize}
The consistency assumption \citep{rubin_estimating_1974} states that $Y=Y_a$ when $A=a$, for $a\in\mathcal A$, and follows from the assumed SCM. Randomized treatment assignment ensures assumptions A1 and A2 in the trial setting considered here. In our setting, we view the subgroup of interest as the population on which we intend to draw inferences, and treat observations from the remainder of the RCT population as auxiliary data ($S=0$) to support estimation.

For an observed-data distribution $P\in\mathcal M$ satisfying subgroup and treatment positivity, define the statistical parameter $\Psi:\mathcal M\to\RR$ by
$$
\Psi(P):=E_P\left[E_P(Y\mid S=1,W,A=1)-E_P(Y\mid S=1,W,A=0)\mid S=1\right].
$$
Under consistency and assumptions A1--A2, the causal parameter is identified by this observed-data functional: $\Psi^F(P_0^F)=\Psi(P_0)$. The definition of $\Psi(P)$ involves only the observed-data distribution; the identification assumptions justify its interpretation as the subgroup causal effect. Throughout the remainder of the paper, our estimators and statistical theory concern $\Psi(P_0)$.

\subsection{Preliminary}

Before introducing specific estimators, we briefly recall two key objects for the construction and analysis of a TMLE in our subgroup setting: the canonical gradient of the target parameter and the exact remainder. The canonical gradient of a pathwise differentiable parameter $\Psi:\mathcal{M}\to\RR$ determines the semi-parametric efficiency bound and, from a practical perspective, identifies the score equation that the targeting step is designed to solve. Once this gradient is available, it immediately suggests the corresponding clever covariate of a TMLE and thus the fluctuation submodel used to update an initial estimate of the relevant factors of the data likelihood. For this reason, our starting point is to derive the canonical gradient of the nonparametric subgroup target itself. The second object is the exact remainder term. For a distribution $P\in\mathcal{M}$, let $D_{\Psi,P}$ denote the canonical gradient of $\Psi$ at $P$, and define the exact remainder $R_{\Psi}(P,P_0)\coloneq \Psi(P)-\Psi(P_0)+P_0D_{\Psi,P}$. If $P_n^\star$ is a targeted plug-in estimator such that $P_nD_{\Psi,P_n^\star}=o_P(n^{-1/2})$, then we have the expansion
$$
\Psi(P_n^\star)-\Psi(P_0)=(P_n-P_0)D_{\Psi,P_n^\star}+R_{\Psi}(P_n^\star,P_0)+o_P(n^{-1/2}).
$$
Equivalently,
$$
\Psi(P_n^\star)-\Psi(P_0)-P_nD_{\Psi,P_0}=R_{\Psi}(P_n^\star,P_0)+(P_n-P_0)\{D_{\Psi,P_n^\star}-D_{\Psi,P_0}\}-P_nD_{\Psi,P_n^\star}.
$$
Thus, once the targeting step makes the empirical mean of the canonical gradient negligible, and the empirical process term is controlled either by Donsker conditions \citep{vaart_weak_1997} or by cross-fitting \citep{zheng_asymptotic_2010}, the main rate condition for asymptotic linearity reduces to showing that the exact remainder is $o_P(n^{-1/2})$. After introducing each estimator, we show the analytic form of the exact remainders and conditions for them to be $\sqrt{n}$-negligible.

\section{TMLE with pooled regression for subgroup analysis}\label{sec:tmle_pr}

As a first estimator, we consider a standard TMLE for the nonparametric subgroup target $\Psi$ whose initial outcome regression is fitted on the full trial sample and then evaluated at $S=1$. We refer to this estimator as TMLE with pooled regression (TMLE-PR). The scientific target remains the same subgroup-specific estimand $\Psi$. The borrowing occurs only through the initial outcome regression fit.

\subsection{Canonical gradient of the target parameter}

For any observed-data distribution $P\in\mathcal{M}$, define $Q_P(s,w,a):=E_P(Y\mid S=s,W=w,A=a)$, $g_P(a\mid s,w):=P(A=a\mid S=s,W=w)$. Thus, $Q_P$ is the conditional outcome regression and $g_P$ is the conditional treatment-assignment mechanism. In the balanced randomized trial considered here, $g_0(a\mid s,w)=1/2$ for $a\in\{0,1\}$. Table \ref{tab:notation} in Appendix \ref{app:notation} summarizes the notation used throughout the paper.

\begin{lemma}\label{lem:can_grad_tmle_pr}
The canonical gradient of the nonparametrically defined target estimand $\Psi$ at $P\in\mathcal{M}$ is 
\begin{align*}
D_{\Psi,P}(O)&=\frac{S}{P(S=1)}\bigl[Q_P(1,W,1)-Q_P(1,W,0)-\Psi(P)\bigr]\\
&+\frac{S}{P(S=1)}\left\{\frac{I(A=1)}{g_P(1\mid 1,W)}-\frac{I(A=0)}{g_P(0\mid 1,W)}\right\}\bigl[Y-Q_P(1,W,A)\bigr].
\end{align*}
\end{lemma}

The derivation can be found in \cite{vanderlaan2024atmle}. This canonical gradient has a particularly favorable form in the RCT setting. Because treatment is randomized at baseline and subgroup membership $S$ is a pre-treatment variable, we still have $A\perp W\mid S=1$, and in fact under complete randomization, $A\perp (S,W)$. Therefore, $g_P(a\mid 1,W)=P(A=a\mid S=1,W)=P(A=a)$, so the treatment mechanism within the subgroup is known by design and does not vary with $W$. In the balanced two-arm trial considered here, $g_P(1\mid 1,W)=g_P(0\mid 1,W)=1/2$, and hence the inverse-weighting part of the canonical gradient is uniformly bounded and does not involve the kind of extreme propensity-score weighting that can arise in observational studies. The main source of difficulty is the factor $1/P(S=1)$, which reflects the fact that inference is restricted to a possibly small subgroup, together with the need to estimate the subgroup-specific outcome regression $Q_P(1,W,A)$ well.

\subsection{Constructing the TMLE-PR estimator}

To construct TMLE-PR, we first obtain an initial estimate $Q_n(s,w,a)$ of $Q_0(s,w,a)=E_0(Y\mid S=s,W=w,A=a)$ using the full RCT sample, with $S$ included as a predictor, and then evaluate this regression at $s=1$. The treatment mechanism $g_0$ is known by randomization and can therefore be plugged in directly (or estimated, which may improve efficiency further \citep{hirano_efficient_2003,balzer_adaptive_2016}). For $P_0(S=1)$, we can take the empirical mean of $S$, i.e., $P_n(S=1)$. Restricting the targeting step to the subgroup observations, we define the clever covariate 
$$
H_n(W,A)\coloneq
\frac{1}{P_n(S=1)}
\left\{
\frac{I(A=1)}{g_n(1\mid 1,W)}-\frac{I(A=0)}{g_n(0\mid 1,W)}
\right\}.
$$
For binary or bounded continuous outcomes, a standard logistic fluctuation submodel for the subgroup part of the regression is
$$
\text{logit}\,Q_{n,\epsilon}(1,W,A)=\text{logit}\,Q_n(1,W,A)+\epsilon H_n(W,A).
$$
The fluctuation parameter $\epsilon_n$ is estimated by logistic regression of $Y$ on $H_n$ using only the subgroup sample and with offset $\text{logit}\,Q_n(1,W,A)$. Let $Q_n^\star$ denote the updated regression. The resulting TMLE-PR estimator is then
$$
\Psi_n^{\mathrm{TMLE\mbox{-}PR}}=\frac{P_n\Bigl[S\bigl\{Q_n^\star(1,W,1)-Q_n^\star(1,W,0)\bigr\}\Bigr]}{P_n(S=1)}.
$$
Inference can be based on the empirical variance of the efficient influence curve in Lemma \ref{lem:can_grad_tmle_pr} evaluated at the targeted estimate.

\begin{remark}\label{rem:subgroup_tmle}
If the initial regression $Q_n$ were instead fitted using only the subgroup observations, then the same targeting construction reduces to an ordinary subgroup-only TMLE. Thus, TMLE-PR differs from the subgroup-only TMLE only in the initial regression step, making it a particularly simple benchmark for us to consider in simulation studies later when assessing how much one gains from pooled outcome-regression borrowing alone.
\end{remark}

We summarize the implementation of TMLE-PR in Algorithm \ref{alg:tmle_pr}.

\begin{breakablealgorithm}
\footnotesize
\caption{Pseudo-code for TMLE with pooled regression (TMLE-PR)}
\label{alg:tmle_pr}
\begin{algorithmic}[1]
\STATE Obtain an initial estimator $Q_n(s,w,a)$ of $Q_0(s,w,a)=E_0(Y\mid S=s,W=w,A=a)$ using the full RCT sample, with $S$ included as a predictor;
\STATE Set the treatment mechanism $g_n(a\mid 1,W)$ equal to the known randomization probability by design (or estimate it);
\STATE Define the clever covariate on the subgroup sample:
$$
H_n(W,A)\coloneq
\frac{1}{P_n(S=1)}
\left\{
\frac{I(A=1)}{g_n(1\mid 1,W)}-\frac{I(A=0)}{g_n(0\mid 1,W)}
\right\};
$$
\STATE Fit a univariate logistic regression using only the observations with $S=1$, with outcome $Y$, offset $\text{logit}\,Q_n(1,W,A)$, and covariate $H_n(W,A)$;
\STATE Let $\epsilon_n$ denote the fitted fluctuation coefficient, and update the subgroup part of the regression:
$$
\text{logit}\,Q_n^\star(1,W,A)
=
\text{logit}\,Q_n(1,W,A)+\epsilon_n H_n(W,A);
$$
for $s=0$, leave the regression unchanged, i.e., $Q_n^\star(0,W,A)=Q_n(0,W,A)$;
\STATE Compute the plug-in TMLE-PR estimator:
$$
\Psi_n^{\mathrm{TMLE\mbox{-}PR}}
=
\frac{
P_n\Bigl[
S\bigl\{Q_n^\star(1,W,1)-Q_n^\star(1,W,0)\bigr\}
\Bigr]
}{
P_n(S=1)
};
$$
\STATE Evaluate the estimated efficient influence curve
\begin{align*}
D_{\Psi,n}^\star(O)
&=\frac{S}{P_n(S=1)}\bigl[Q_n^\star(1,W,1)-Q_n^\star(1,W,0)-\Psi_n^{\mathrm{TMLE\mbox{-}PR}}\bigr]\\
&+\frac{S}{P_n(S=1)}
\left\{\frac{I(A=1)}{g_n(1\mid 1,W)}-\frac{I(A=0)}{g_n(0\mid 1,W)}\right\}\bigl[Y-Q_n^\star(1,W,A)\bigr];
\end{align*}
\STATE Compute the Wald-type 95\% confidence interval
$$
\Psi_n^{\mathrm{TMLE\mbox{-}PR}}
\pm
1.96\sqrt{P_n\bigl\{(D_{\Psi,n}^\star)^2\bigr\}/n}.
$$
\end{algorithmic}
\end{breakablealgorithm}

\subsection{Exact remainder of the target parameter $\Psi$}

We now study the exact remainder of the nonparametrically defined subgroup target
$$
\Psi(P)=E_P\bigl[Q_P(1,W,1)-Q_P(1,W,0)\mid S=1\bigr].
$$
Define $m_P(W)\coloneq Q_P(1,W,1)-Q_P(1,W,0)$,  $p_P\coloneq P(S=1)$, and $p_0\coloneq P_0(S=1)$. Let $N_\Psi(P)\coloneq E_P\bigl[S\,m_P(W)\bigr]$. Then, $\Psi(P)=N_\Psi(P)/p_P=E_P[m_P(W)\mid S=1]$.

\begin{lemma}\label{lem:ratio_remainder_conditional_mean}
Let $h_P(W)$ be a $P$-dependent measurable function of $W$ and define
$$
\Phi_h(P):=E_P[h_P(W)\mid S=1]=\frac{N_h(P)}{p_P}
$$
and $N_h(P):=E_P[S\,h_P(W)]$. Suppose $D_{N_h,P}$ is the canonical gradient of the numerator parameter $N_h(P)$ at $P$, and define its exact remainder $R_{N_h}(P,P_0):=N_h(P)-N_h(P_0)+P_0D_{N_h,P}$. Then, the exact remainder of $\Phi_h$ at $P$ relative to $P_0$ is
\begin{align*}
R_{\Phi_h}(P,P_0)
&:=\Phi_h(P)-\Phi_h(P_0)+P_0D_{\Phi_h,P}\\
&=\frac{1}{p_P}R_{N_h}(P,P_0)+\frac{p_P-p_0}{p_P}\{\Phi_h(P)-\Phi_h(P_0)\},
\end{align*}
where
$$
D_{\Phi_h,P}(O)=\frac{D_{N_h,P}(O)-\Phi_h(P)\{S-p_P\}}{p_P}.
$$
\end{lemma}

\begin{lemma}\label{lem:subgroup_remainder_np_tmle}
The exact remainder $R_\Psi(P,P_0)\coloneq \Psi(P)-\Psi(P_0)+P_0D_{\Psi,P}$ for the nonparametrically defined target parameter $\Psi$ is
\begin{align*}
R_\Psi(P,P_0)&=\frac{1}{p_P}E_0\left[S\cdot\frac{g_P-g_0}{g_P}(1\mid 1,W)\bigl\{Q_P(1,W,1)-Q_0(1,W,1)\bigr\}
\right]\\
&-\frac{1}{p_P}E_0\left[S\cdot\frac{g_P-g_0}{g_P}(0\mid 1,W)\bigl\{Q_P(1,W,0)-Q_0(1,W,0)\bigr\}\right]\\
&+\frac{p_P-p_0}{p_P}\{\Psi(P)-\Psi(P_0)\}.
\end{align*}
\end{lemma}

The exact remainder in this subgroup RCT setting is particularly favorable. First, if the treatment mechanism is known by design and correctly plugged in, i.e., $g_P(a\mid 1,W)=g_0(a\mid 1,W)$ for $a\in\{0,1\}$, then the numerator remainder vanishes exactly: $R_{N_\Psi}(P,P_0)=0$, regardless of misspecification of $Q_P(1,W,a)$. What remains is only the term
$$
\frac{p_P-p_0}{p_P}\{\Psi(P)-\Psi(P_0)\}.
$$
Since $p_P=P(S=1)$ is a one-dimensional mean parameter, it is estimated at the parametric rate $n^{-1/2}$. Therefore, this denominator term is $o_P(n^{-1/2})$ provided $\Psi(P_n^\star)-\Psi(P_0)=o_P(1)$. If $p_P=p_0$ is correctly specified, this denominator term vanishes exactly. So, in this RCT subgroup-analysis setting, we generally expect the exact remainder of TMLE-PR to be small, and its numerator component is in fact exactly zero once the randomized treatment mechanism is correctly specified.

\section{Adaptive TMLE for subgroup analysis}\label{sec:atmle}

Adaptive TMLE (a.k.a ``adaptive debiased machine learning'') is a statistical estimation framework originally proposed in \cite{van2023adaptive} to improve the stability of TMLE in challenging scenarios, e.g., settings with limited overlap in covariates between the treatment and control groups. A-TMLE regularizes the targeting step in TMLE and often results in more stable variance estimates in finite samples in those scenarios. A version of A-TMLE was developed in the context of RCT-RWD integration to improve efficiency \citep{vanderlaan2024atmle}. We provide a brief review of the key elements of this A-TMLE and adapt them to the subgroup analysis setting.

\subsection{Decomposition of the estimand and its corresponding projection estimand}

In the A-TMLE for RCT-RWD integration, one first decomposes the target estimand into a pooled estimand and a bias estimand \citep{vanderlaan2024atmle}. Here, we follow the same decomposition strategy. For any probability distribution $P\in\mathcal{M}$, define
$$
Q_P(s,w,a):=E_P(Y\mid S=s,W=w,A=a),\quad
\bar Q_P(w,a):=E_P(Y\mid W=w,A=a),
$$
and let
$$
\Pi_P(s\mid w,a):=P(S=s\mid W=w,A=a),\quad
\tau_{S,P}(w,a):=Q_P(1,w,a)-Q_P(0,w,a).
$$
We first define the pooled estimand $\tilde\Psi(P)
:=
E_{P_{W\mid S=1}}
[\bar Q_P(W,1)-\bar Q_P(W,0)]$. This estimand averages over the covariate distribution of the target subgroup, but uses the outcome regression learned from the pooled sample. This estimand is a more efficient target from an estimation perspective because the inner outcome regression allows us to leverage both $S=1$ and $S=0$ (i.e., the entire RCT data) to estimate. However, for
example, when the subgroup population has a different distribution of unmeasured baseline covariates than the rest of the population, and these covariates modify the treatment effect on an absolute scale, or if the variable used to define the subgroup itself modifies the treatment effect directly, then the estimand $\tilde{\Psi}$ may not necessarily equal our target of interest, $\Psi$. To capture this difference, we define the bias estimand as $\Psi^\#(P)=\tilde{\Psi}(P)-\Psi(P)$, which captures precisely the difference between the pooled estimand and our target estimand. It has been shown in \cite{vanderlaan2024atmle} (replacing the outer expectation $P_W$ with $P_{W\mid S=1}$ in the proof of Lemma 2) that the bias estimand $\Psi^\#$ has an analytic form:
$$
\Psi^\#(P):=E_{P_{W\mid S=1}}[\Pi_P(0\mid W,0)\tau_{S,P}(W,0)-\Pi_P(0\mid W,1)\tau_{S,P}(W,1)],
$$
a combination of the bias function $\tau_{S,P}$ weighted by the subgroup enrollment mechanism $\Pi_P$. Now, in order to obtain an estimator for $\Psi$, one may obtain estimators for $\tilde{\Psi}$ and $\Psi^\#$, respectively, and take the difference.

A key feature of A-TMLE is that instead of directly estimating nonparametrically the estimands $\tilde{\Psi}$ and $\Psi^\#$, it estimates the so-called projection estimands. Even though it does not directly estimate the nonparametrically defined target estimand, valid inference is still possible because the difference between the projection estimand and the target estimand is shown to be second-order under certain conditions \citep{van2023adaptive}. We now formally define the projection estimands. Let $\mathcal T_{A,w}=\{\tau_{A,\beta}:\beta\in\mathcal{B}_A\}$ be a working model for the pooled conditional average treatment effect (CATE) $\tau_{A,P}(w):=\bar Q_P(w,1)-\bar Q_P(w,0)$, and let the associated semi-parametric regression working model be $\bar{\mathcal Q}_w
=
\bigl\{\theta(w)+a\,\tau_{A,\beta}(w): \theta,\beta\bigr\}$, where $\theta(\cdot)$ is some unspecified function of $w$. Here, the subscript $w$ denotes ``working model''; throughout the article, any appearance of $w$ in the subscript indicates that the corresponding object is associated with a working model. Let $\bar{g}_P(a\mid W):=P(A=a\mid W)$. Define $\tau_{A,\beta_A(P)}$ as the weighted squared-error projection of $\tau_{A,P}$ onto $\mathcal T_{A,w}$:
$$
\beta_A(P)\in
\arg\min_{\beta}
E_P\left[
\bar{g}_P(1\mid W)\{1-\bar{g}_P(1\mid W)\}
\bigl\{\tau_{A,P}(W)-\tau_{A,\beta}(W)\bigr\}^2
\right].
$$
Equivalently, this is the projection induced by the squared-error projection of $\bar Q_P$ onto $\bar{\mathcal Q}_w$ \citep{crump_moving_2006,li_addressing_2019,damour_overlap_2021,morzywolek_on_2023,van2023adaptive}. The corresponding projection parameter for $\tilde{\Psi}$ is then
$$
\tilde\Psi_{\mathcal M_{A,w}}(P)
:=
E_{P_{W\mid S=1}}
\bigl[\tau_{A,\beta_A(P)}(W)\bigr].
$$
Similarly, let $\mathcal T_{S,w}=\{\tau_{S,\beta}:\beta\in\mathcal B_S\}$ be a working model for the conditional average subgroup effect $\tau_{S,P}(W,A)$, with associated semi-parametric regression working model $\mathcal Q_w
=
\bigl\{\theta(W,A)+S\,\tau_{S,\beta}(W,A): \theta,\beta\bigr\}$, where $\theta(\cdot,\cdot)$ is now an unspecified function of $(w,a)$. Define $\tau_{S,\beta_S(P)}$ as the weighted squared-error projection of $\tau_{S,P}$ onto $\mathcal T_{S,w}$:
$$
\beta_S(P)\in
\arg\min_{\beta}
E_P\left[
\Pi_P(1\mid W,A)\{1-\Pi_P(1\mid W,A)\}
\bigl\{\tau_{S,P}(W,A)-\tau_{S,\beta}(W,A)\bigr\}^2
\right].
$$
Equivalently, this is induced by the squared-error projection of $Q_P$ onto $\mathcal Q_w$. The corresponding projection parameter for $\Psi^\#$ is then
$$
\Psi^\#_{\mathcal M_{S,w}}(P)
:=
E_{P_{W\mid S=1}}
\Bigl[
\Pi_P(0\mid W,0)\tau_{S,\beta_S(P)}(W,0)
-
\Pi_P(0\mid W,1)\tau_{S,\beta_S(P)}(W,1)
\Bigr].
$$
We therefore obtain the projected version of our target parameter:
$$
\Psi_{\mathcal{M}_w}(P)
:=
\tilde\Psi_{\mathcal M_{A,w}}(P)-\Psi^\#_{\mathcal M_{S,w}}(P),
$$
where $\mathcal{M}_{A,w}\coloneq\{P\in\mathcal{M}:\tau_{A,P}\in\mathcal{T}_{A,w}\}$, $\mathcal{M}_{S,w}\coloneq\{P\in\mathcal{M}:\tau_{S,P}\in\mathcal{T}_{S,w}\}$, and $\mathcal{M}_w\coloneq\mathcal{M}_{A,w}\cap\mathcal{M}_{S,w}$.

\subsection{Canonical gradient of the projection estimand $\Psi_{\mathcal{M}_w}$ at $P\in\mathcal{M}_w$}

We consider the canonical gradient of $\tilde{\Psi}_{\mathcal{M}_{A,w}}$ (Lemma \ref{lem:can_grad_psi_tilde}) and $\Psi^\#_{\mathcal{M}_{S,w}}$ (Lemma \ref{lem:can_grad_psi_pound}), respectively. Their derivations are provided in Appendix \ref{app:can_grads}.
\begin{lemma}\label{lem:can_grad_psi_tilde}
We define
\begin{gather*}
I_{A,P}=E_P\Big[\bar{g}_P(1-\bar{g}_P)(1\mid W)\phi_A\phi_A^\top(W)\Big]\quad\text{and}\\
\bar{Q}_{w,P}(W,A)=E_P(Y\mid W)+(A-\bar{g}_P(1\mid W))\sum_j\beta_{A,P}(j)\phi_{A,j}(W).
\end{gather*}
The canonical gradient of $\tilde{\Psi}_{\mathcal{M}_{A,w}}$ at $P$ is given by
$$
D_{\tilde{\Psi},\mathcal{M}_{A,w},P}=\frac{S}{P(S=1)}\left[\tau_{A,\beta_{A,P}}(W)-\tilde{\Psi}_{\mathcal{M}_{A,w}}(P)\right]+\frac{1}{P(S=1)}\Big\{E_P[S\phi_A(W)]\Big\}^\top D_{A,\beta_A,P},
$$
where
$$
D_{A,\beta_A,P}=I_{A,P}^{-1}(A-\bar{g}_P(1\mid W))\phi_A(W)(Y-\bar{Q}_{w,P}(W,A)).
$$
\end{lemma}

\begin{lemma}\label{lem:can_grad_psi_pound}
Let $\bar{\Pi}_P(1\mid W)\coloneq P(S=1\mid W)$. We define
\begin{align*}
I_{S,P}&=E_P\Big[\Pi_P(1-\Pi_P)(1\mid W,A)\phi_S\phi_S^\top(W,A)\Big]\quad\text{and}\\
Q_{w,P}(S,W,A)&=E_P(Y\mid W,A)+(S-\Pi_P(1\mid W,A))\sum_j\beta_{S,P}(j)\phi_{S,j}(W,A).
\end{align*}
The canonical gradient of the bias projection parameter $\Psi^\#_{\mathcal{M}_{S,w}}$ at $P$ is given by
$$
D_{\Psi^\#,\mathcal{M}_{S,w},P}=D_{\Psi^\#,\mathcal{M}_{S,w},P_W,P}+D_{\Psi^\#,\mathcal{M}_{S,w},\Pi,P}+D_{\Psi^\#,\mathcal{M}_{S,w},\beta_S,P},
$$
where
\begin{align*}
D_{\Psi^\#,\mathcal{M}_{S,w},P_W,P}&=\frac{S}{P(S=1)}\bigg[\Pi_P(0\mid W,0)\tau_{S,\beta_{S,P}}(W,0)-\Pi_P(0\mid W,1)\tau_{S,\beta_{S,P}}(W,1)-\Psi_{\mathcal{M}_{S,w}}^\#(P)\bigg],\\
D_{\Psi^\#,\mathcal{M}_{S,w},\Pi,P}&=\frac{\bar{\Pi}_P(1\mid W)}{P(S=1)}\left[\frac{A}{\bar{g}_P(1\mid W)}\tau_{S,\beta_{S,P}}(W,1)-\frac{1-A}{\bar{g}_P(0\mid W)}\tau_{S,\beta_{S,P}}(W,0)\right](S-\Pi_P(1\mid W,A)),\\
D_{\Psi^\#,\mathcal{M}_{S,w},\beta,P}&=\frac{1}{P(S=1)}\bigg[E_P[S\cdot\Pi_P(0\mid W,0)\phi_S(W,0)]-E_P[S\cdot\Pi_P(0\mid W,1)\phi_S(W,1)]\bigg]^\top D_{S,\beta_S,P}.
\end{align*}
where
$$
D_{S,\beta_S,P}=I_{S,P}^{-1}(S-\Pi_P(1\mid W,A))\phi_{S}(W,A)(Y-Q_{w,P}(S,W,A)).
$$
\end{lemma}

\subsection{Understanding the canonical gradients in the context of subgroup analysis}

We now study components of the canonical gradient for the bias projection parameter
$\Psi^\#_{\mathcal{M}_{S,w}}$ to gain some insights on the behavior of the estimator in our RCT subgroup analysis setting. Under complete randomization of treatment, we have $A \perp (S,W)$, and therefore $\Pi_P(s\mid W,a)=P(S=s\mid W)$ and $\bar g_P(a\mid W)=P(A=a)$, for $s,a\in\{0,1\}$. In the balanced two-arm trial considered here,
$P_0(A=1)=P_0(A=0)=1/2$. Under this randomization, the bias projection parameter simplifies to
$$
\Psi^\#_{\mathcal{M}_{S,w}}(P)=E_{P_{W\mid S=1}}
\Big[
\bar{\Pi}_P(0\mid W)\underbrace{\big\{
\tau_{S,\beta_{S,P}}(W,0)-\tau_{S,\beta_{S,P}}(W,1)
\big\}}_{\text{subgroup-by-treatment interaction}}
\Big].
$$
This representation is important. It shows that, in a randomized trial, the bias correction is driven by subgroup-by-treatment interaction, not by subgroup differences in the outcome level that are common to both treatment arms. In particular, if $\tau_{S,\beta_{S,P}}(W,0)=\tau_{S,\beta_{S,P}}(W,1)$ for all $W$, then $\Psi^\#_{\mathcal{M}_{S,w}}(P)=0$. Thus, when subgroup membership affects the outcome equally in the treatment and control arms, borrowing strength from $S=0$ does not induce bias for the subgroup-specific treatment effect, in which case we can fully rely on $\tilde{\Psi}=\Psi$ and leverage the full RCT sample for estimation, allowing us to achieve, potentially, a considerable amount of efficiency gain. On the other hand, when the working model for $\tau_{S,P}$ is selected in a data-adaptive manner from a collection of candidate models, some of which exclude the treatment variable $A$, there is a non-zero probability that the selected model yields a zero bias estimand even when there is a lack of knowledge for that. This may not be desirable. To mitigate this issue, we restrict model selection to a class of working models that all include the treatment variable $A$. In practice, this can be implemented by, for example, leaving the coefficient of $A$ unpenalized when fitting a Lasso model (and similarly for the Highly Adaptive Lasso, as discussed in the following section).

Let's now interpret each of the three components of the canonical gradient of the bias projection parameter respectively. First, for the $P_W$-component, randomization yields
$$
D_{\Psi^\#,\mathcal{M}_{S,w},P_W,P}=\frac{S}{P(S=1)}\Big[\bar{\Pi}_P(0\mid W)\big\{\tau_{S,\beta_{S,P}}(W,0)-\tau_{S,\beta_{S,P}}(W,1)\big\}-\Psi^\#_{\mathcal{M}_{S,w}}(P)\Big].
$$
Therefore, this term is simply the centered version of the weighted function
$$
\bar{\Pi}_P(0\mid W)\big\{\tau_{S,\beta_{S,P}}(W,0)-\tau_{S,\beta_{S,P}}(W,1)\big\}.
$$
Its variance is small when this weighted arm-specific contrast is small or has limited
variability. The important quantity here is the contrast
$\tau_{S,\beta_{S,P}}(W,0)-\tau_{S,\beta_{S,P}}(W,1)$, not the overall magnitude of the arm-specific bias function $\tau_{S,\beta_{S,P}}(W,a)$ by itself. Second, for the $\Pi$-component, randomization implies
$$
D_{\Psi^\#,\mathcal{M}_{S,w},\Pi,P}
=\frac{\bar{\Pi}_P(1\mid W)}{P(S=1)}
\left[
\frac{A}{g_P(1\mid W)}\tau_{S,\beta_{S,P}}(W,1)-\frac{1-A}{g_P(0\mid W)}\tau_{S,\beta_{S,P}}(W,0)
\right](S-\bar{\Pi}_P(1\mid W)).
$$
In the balanced randomization case, this becomes
$$
D_{\Psi^\#,\mathcal{M}_{S,w},\Pi,P}=\frac{\bar{\Pi}_P(1\mid W)}{P(S=1)}
\left[
2A\,\tau_{S,\beta_{S,P}}(W,1)-2(1-A)\,\tau_{S,\beta_{S,P}}(W,0)\right]
(S-\bar{\Pi}_P(1\mid W)).
$$
This term reflects the contribution of estimating the subgroup-membership mechanism in the larger semi-parametric model. Its variance is controlled when the arm-specific subgroup-effect functions are small or well-behaved. At the same time, the factor $1/P(S=1)$ shows that very small subgroup prevalence will still inflate variance, even in this favorable randomized-trial setting. Third, for the $\beta_S$-component, randomization gives
$$
D_{\Psi^\#,\mathcal{M}_{S,w},\beta,P}=\frac{1}{P(S=1)}
\Bigg[
E_P\Big\{
S\,\bar{\Pi}_P(0\mid W)
\big(
\phi_S(W,0)-\phi_S(W,1)
\big)
\Big\}
\Bigg]^\top
D_{S,\beta_S,P}.
$$
Hence, only those basis directions that change with treatment contribute to this term. Equivalently, the size of this component is governed by the treatment-dependent part of the working model for $\tau_{S,\beta_S}$. If the working model contains no $W\times A$ interactions, then all basis functions depending only on $W$ cancel in $\phi_S(W,0)-\phi_S(W,1)$, and only basis functions involving $A$ alone remain. If the working model does not depend on $A$ at all, then $\phi_S(W,0)=\phi_S(W,1)$ and $D_{\Psi^\#,\mathcal{M}_{S,w},\beta,P}=0$. In that special case, the bias projection parameter is also zero, which is consistent with the fact that there is no subgroup-specific treatment-effect bias to correct for.

These simplifications show why subgroup analysis within a randomized trial can be a particularly favorable setting for A-TMLE. The bias correction depends only on treatment-specific subgroup differences, rather than on arbitrary subgroup differences in the outcome regression. This is intuitive because only treatment-specific subgroup differences would result in the CATE differing in the subgroup versus the full population. Consequently, when subgroup-by-treatment interaction is weak, or can be represented by a low-complexity working model, the contribution of the bias term to the asymptotic variance can be small. However, this should be interpreted together with the unavoidable factor $1/P(S=1)$. That is, if the subgroup is rare, variance inflation remains unavoidable.

\begin{remark}
We note one subtle but important point. The display above is obtained by simplifying the canonical gradient derived in the larger model of Lemma \ref{lem:can_grad_psi_pound}. If one further restricts the statistical model to the randomized-trial submodel in which $\Pi_P(s\mid W,a)=P(S=s\mid W)$ by design, then the efficient influence function may simplify further. Thus, plugging the randomization restriction into the displayed canonical gradient is highly informative for interpretation, but it is conceptually distinct from re-deriving the efficient influence function in the smaller randomized-trial model itself.
\end{remark}

\subsection{Asymptotic linearity and efficiency of the subgroup A-TMLE}

The general proof of asymptotic linearity and efficiency of this A-TMLE is provided in \cite{van2023adaptive} and \cite{vanderlaan2024atmle}. Here, for completeness, we summarize again the conditions under which the current subgroup A-TMLE is asymptotically linear and efficient. The theorem is the direct subgroup analogue of the corresponding result in the RCT-RWD setting, with the exact remainder replaced by the one presented later in Theorem \ref{thm:subgroup_exact_remainder}.

\begin{theorem}\label{thm:subgroup_atmle_theory}
Let $D_{\mathcal{M}_w,P}$ be the canonical gradient of $\Psi_{\mathcal{M}_w}$ at $P$, and define $R_{\mathcal{M}_w}(P,P_0):=\Psi_{\mathcal{M}_w}(P)-\Psi_{\mathcal{M}_w}(P_0)+P_0D_{\mathcal{M}_w,P}$. Let ${\cal M}_0\subset {\cal M}$ denote an oracle model containing $P_0$, with canonical gradient $D_{{\cal M}_0,P}$ and exact remainder $R_{{\cal M}_0}(P,P_0)$ for the oracle projection parameter $\Psi_{{\cal M}_0}$. Let $P_n^\star\in{\cal M}_w$ be an estimator satisfying the score equation $P_nD_{\mathcal{M}_w,P_n^\star}=o_P(n^{-1/2})$.
Then, we have
$$
\Psi_{\mathcal{M}_w}(P_n^\star)-\Psi_{\mathcal{M}_w}(P_0)=(P_n-P_0)D_{\mathcal{M}_w,P_n^\star}+R_{\mathcal{M}_w}(P_n^\star,P_0)+o_P(n^{-1/2}).
$$

\paragraph{Oracle model approximation condition.}
Let $P_{0,n}:=\Pi_{{\cal M}_w}(P_0)$ and $\tilde P_{0,n}:=\Pi_{{\cal M}_0}(P_{0,n})$. Let $\tilde D_{{\cal M}_w,\tilde P_{0,n}}$ be the projection of $D_{{\cal M}_0,\tilde P_{0,n}}$ onto the tangent space $T_{{\cal M}_w}(\tilde P_{0,n})$. Assume
\begin{itemize}
\item $\Psi(\tilde P_{0,n})-\Psi(P_{0,n})=o_P(n^{-1/2})$ and
\item $(\tilde P_{0,n}-P_0)\Bigl\{D_{{\cal M}_0,\tilde P_{0,n}}-\tilde D_{{\cal M}_w,\tilde P_{0,n}}\Bigr\}=o_P(n^{-1/2}).$
\end{itemize}
Then, $\Psi_{\mathcal{M}_w}(P_n^\star)-\Psi(P_0)=(P_n-P_0)D_{\mathcal{M}_w,P_n^\star}+R_{\mathcal{M}_w}(P_n^\star,P_0)+o_P(n^{-1/2})$.

\paragraph{Rate condition.}
If $R_{\mathcal{M}_w}(P_n^\star,P_0)=o_P(n^{-1/2})$, then $\Psi_{\mathcal{M}_w}(P_n^\star)-\Psi(P_0)=(P_n-P_0)D_{\mathcal{M}_w,P_n^\star}+o_P(n^{-1/2})$.

\paragraph{Donsker or cross-fitted empirical process condition.}
If either
\begin{itemize}
\item $D_{\mathcal{M}_w,P_n^\star}$ falls with probability tending to one in a $P_0$-Donsker class and $P_0\{D_{\mathcal{M}_w,P_n^\star}-D_{\mathcal{M}_w,P_0}\}^2\rightarrow_p 0$, or
\item a CV-TMLE is used so that the corresponding empirical process term is asymptotically negligible,
\end{itemize}
then $\Psi_{\mathcal{M}_w}(P_n^\star)-\Psi(P_0)
=(P_n-P_0)D_{\mathcal{M}_w,P_0}+o_P(n^{-1/2})$.

\paragraph{Asymptotic normality.}
If, in addition, $n^{1/2}P_nD_{\mathcal{M}_w,P_0}/\sigma_{w,n}\Rightarrow_d N(0,1)$, where $\sigma_{w,n}^2:=P_0D_{\mathcal{M}_w,P_0}^2$, then
\[
\sigma_{w,n}^{-1}n^{1/2}\{\Psi_{\mathcal{M}_w}(P_n^\star)-\Psi(P_0)\}\Rightarrow_d N(0,1).
\]
If the stronger stability condition $\|D_{\mathcal{M}_w,P_0}-D_{{\cal M}_0,P_0}\|_{P_0}=o_P(1)$ holds, then
$$
\Psi_{\mathcal{M}_w}(P_n^\star)-\Psi(P_0)=(P_n-P_0)D_{{\cal M}_0,P_0}+o_P(n^{-1/2}),
$$
so that the subgroup A-TMLE is asymptotically linear with influence curve equal to the efficient influence curve of $\Psi:{\cal M}_0\rightarrow\RR$ at $P_0$.
\end{theorem}


\subsection{Conditions for zero projection bias}
\label{subsec:zero_projection_bias}

Theorem \ref{thm:subgroup_atmle_theory} requires control of the approximation error introduced by the working models when replacing $\Psi_{\mathcal M_w}(P_0)$ by the scientific target $\Psi(P_0)$. We now provide concrete sufficient conditions under which this projection gap is exactly zero in the randomized-trial setting. These conditions do not require the true conditional effect functions to belong to their respective working models. We also give a product-error characterization that yields sufficient conditions for an asymptotically negligible projection gap.

Throughout this subsection, assume that $A\perp(S,W)$ under $P_0$, $P_0(A=1)=1/2$, and $P_0(S=1)>0$. Consider the unpenalized population projections defined above onto the linear working-model spans generated by $\phi_A$ and $\phi_S$. Assume that the required conditional means and expectations are well defined and that $I_{A,0}$ and $I_{S,0}$ are nonsingular. Write $E_0=E_{P_0}$, $\beta_{A,0}=\beta_A(P_0)$, and $\beta_{S,0}=\beta_S(P_0)$, and recall that $\bar\Pi_0(1\mid W)=P_0(S=1\mid W)$.

\begin{lemma}\label{lem:projection_target_preservation}
The following statements hold.
\begin{enumerate}
\item If $\bar\Pi_0(1\mid W)$ belongs to the linear span of $\phi_A(W)$, then
$$
\tilde\Psi_{\mathcal M_{A,w}}(P_0)=\tilde\Psi(P_0).
$$
\item If the linear span of $\phi_S(W,A)$ contains an intercept and the treatment indicator $A$, then
$$
\Psi^\#_{\mathcal M_{S,w}}(P_0)=\Psi^\#(P_0).
$$
\end{enumerate}
Neither statement requires the corresponding true conditional effect function to belong to its working model.
\end{lemma}

\begin{proof}
For the first statement, balanced randomization makes $\bar g_0(1\mid W)\{1-\bar g_0(1\mid W)\}=1/4$. The projection normal equations therefore imply $E_0\left[\phi_A(W)\{\tau_{A,\beta_{A,0}}(W)-\tau_{A,P_0}(W)\}\right]=0$. The projection gap can be written as
\begin{align*}
\tilde\Psi_{\mathcal M_{A,w}}(P_0)-\tilde\Psi(P_0)=\frac{1}{P_0(S=1)}E_0\left[\bar\Pi_0(1\mid W)\{\tau_{A,\beta_{A,0}}(W)-\tau_{A,P_0}(W)\}\right].
\end{align*}
If $\bar\Pi_0(1\mid W)$ is a linear combination of the basis functions, the normal equations make this expectation zero.

For the second statement, the subgroup-effect projection satisfies
\begin{align*}
E_0\Bigl[\Pi_0(1\mid W,A)\Pi_0(0\mid W,A)\phi_S(W,A)\{\tau_{S,\beta_{S,0}}(W,A)-\tau_{S,P_0}(W,A)\}\Bigr]=0.
\end{align*}
Because the span contains an intercept and $A$, it also contains $2-4A$. Moreover, randomization implies $\Pi_0(s\mid W,A)=\bar\Pi_0(s\mid W)$. Conditioning on $W$ and averaging over treatment gives
\begin{align*}
\Psi^\#_{\mathcal M_{S,w}}(P_0)-\Psi^\#(P_0)=\frac{1}{P_0(S=1)}E_0\Bigl[\Pi_0(1\mid W,A)\Pi_0(0\mid W,A)(2-4A)\{\tau_{S,\beta_{S,0}}(W,A)-\tau_{S,P_0}(W,A)\}\Bigr]=0,
\end{align*}
where the last equality follows from the projection normal equations.
\end{proof}

\begin{corollary}\label{cor:zero_subgroup_projection_bias}
If both conditions in Lemma \ref{lem:projection_target_preservation} hold, then
$$
\Psi_{\mathcal M_w}(P_0)=\Psi(P_0).
$$
More generally, if the subgroup-effect working-model span contains an intercept and $A$, then
$$
\Psi_{\mathcal M_w}(P_0)-\Psi(P_0)=\tilde\Psi_{\mathcal M_{A,w}}(P_0)-\tilde\Psi(P_0).
$$
\end{corollary}

\begin{proof}
Subtract the two component projection gaps in the decomposition $\Psi_{\mathcal M_w}=\tilde\Psi_{\mathcal M_{A,w}}-\Psi^\#_{\mathcal M_{S,w}}$ and apply Lemma \ref{lem:projection_target_preservation}.
\end{proof}

These conditions clarify which averages are preserved by projection. An intercept in the pooled-effect working model preserves the pooled average effect, whereas inclusion of $\bar\Pi_0(1\mid W)$ preserves the subgroup-standardized average. For the bias component, retaining an intercept and $A$ preserves the integrated bias parameter, even if the subgroup-effect working model omits more complex interactions.

\paragraph{Using an estimated subgroup-membership function.}

The true function $\bar\Pi_0(1\mid W)$ is generally unknown. Suppose instead that the selected pooled-effect span contains an estimate $\bar\Pi_n(1\mid W)$, while the subgroup-effect span contains an intercept and $A$. For each realized working-model span, the same normal equations imply
\begin{align*}
\Psi_{\mathcal M_w}(P_0)-\Psi(P_0)=\frac{1}{P_0(S=1)}E_0\Bigl[(\bar\Pi_0-\bar\Pi_n)(1\mid W)\{\tau_{A,\beta_{A,0}}(W)-\tau_{A,P_0}(W)\}\Bigr].
\end{align*}
Consequently, by the Cauchy--Schwarz inequality,
\begin{align*}
\left|\Psi_{\mathcal M_w}(P_0)-\Psi(P_0)\right|\leq\frac{\|(\bar\Pi_n-\bar\Pi_0)(1\mid\cdot)\|_{P_0}}{P_0(S=1)}\|\tau_{A,\beta_{A,0}}-\tau_{A,P_0}\|_{P_0},
\end{align*}
where the norms are taken in $L^2(P_{0,W})$. Thus, including an estimated membership function does not generally make the gap exactly zero, but gives a sufficient condition for an $o_P(n^{-1/2})$ gap when the product of these two errors is $o_P(n^{-1/2})$. The population projection is defined on the selected linear span; the normal equations do not automatically hold for a penalized fit.

Under Corollary \ref{cor:zero_subgroup_projection_bias}, the replacement of $\Psi_{\mathcal M_w}(P_0)$ by $\Psi(P_0)$ in Theorem \ref{thm:subgroup_atmle_theory} is exact. The product-error condition above gives an alternative sufficient condition for making that replacement at the $n^{-1/2}$ scale. Both conclusions concern the projection gap at $P_0$. Control of the estimation remainder and empirical-process terms, and consistent variance estimation, remain necessary for inference.

\subsection{Exact remainder of the projection parameter $\Psi_{\mathcal{M}_w}$}

The main condition in Theorem \ref{thm:subgroup_atmle_theory} is the rate condition $R_{\mathcal{M}_w}(P_n^\star,P_0)=o_P(n^{-1/2})$. We therefore also study the exact remainder of the projection parameter $\Psi_{\mathcal{M}_w}(P)=\tilde{\Psi}_{\mathcal{M}_{A,w}}(P)-\Psi^\#_{\mathcal{M}_{S,w}}(P)$, using the gradient expressions in Lemmas \ref{lem:can_grad_psi_tilde} and \ref{lem:can_grad_psi_pound}. The calculations use the larger-model gradients described above. Proofs are deferred to Appendix \ref{app:exact_remainder}.

Throughout, fix the same linear working bases at $P$ and $P_0$, and assume that the information matrices are nonsingular. Assume also that the conditional functions are defined on the relevant $P_0$ support, that the treatment denominators are positive, and that all displayed expectations are finite. Write $E_0=E_{P_0}$, $p_P=P(S=1)$, and $p_0=P_0(S=1)$, with $p_P,p_0>0$. The subscript $0$ denotes evaluation at $P_0$; in particular, $\beta_{A,0}$ and $\beta_{S,0}$ are projection coefficients at $P_0$. The true effect functions need not belong to the working models. As before, $\theta_P(W)=E_P(Y\mid W)$, $\bar Q_P(W,A)=E_P(Y\mid W,A)$, and $\bar\Pi_P(1\mid W)=P(S=1\mid W)$. Define $m_{A,P}(W)=\tau_{A,\beta_{A,P}}(W)$, $b_P(W)=\Pi_P(0\mid W,0)\tau_{S,\beta_{S,P}}(W,0)-\Pi_P(0\mid W,1)\tau_{S,\beta_{S,P}}(W,1)$, and $m_P(W)=m_{A,P}(W)-b_P(W)$. Then $\tilde{\Psi}_{\mathcal{M}_{A,w}}(P)=E_P[m_{A,P}(W)\mid S=1]$, $\Psi^\#_{\mathcal{M}_{S,w}}(P)=E_P[b_P(W)\mid S=1]$, and $\Psi_{\mathcal{M}_w}(P)=E_P[m_P(W)\mid S=1]$.

Subgroup standardization changes the numerator remainder as well as introducing the denominator remainder. We first consider the numerator $N_A(P)=E_P[S\,\tau_{A,\beta_{A,P}}(W)]$.

\begin{lemma}\label{lem:subgroup_remainder_tilde_numerator}
The exact remainder $R_{N_A}(P,P_0):=N_A(P)-N_A(P_0)+P_0D_{N_A,P}$ is given by
\begin{align*}
R_{N_A}(P,P_0)&=\Bigl\{E_0[S\phi_A(W)]^\top-E_P[S\phi_A(W)]^\top I_{A,P}^{-1}I_{A,0}\Bigr\}(\beta_{A,P}-\beta_{A,0})\\
&+E_P[S\phi_A(W)]^\top I_{A,P}^{-1}E_0\Bigl[\phi_A(W)\Bigl\{(\bar g_P-\bar g_0)(1\mid W)(\theta_P-\theta_0)(W)\\
&-(\bar g_P-\bar g_0)^2(1\mid W)\tau_{A,\beta_{A,P}}(W)\Bigr\}\Bigr].
\end{align*}
\end{lemma}

\begin{corollary}\label{cor:subgroup_remainder_tilde_known_g}
If $\bar g_P=\bar g_0$, then
\begin{align*}
R_{N_A}(P,P_0)
&=\Bigl\{E_0[S\phi_A(W)]^\top-E_P[S\phi_A(W)]^\top I_{A,P}^{-1}I_{A,0}\Bigr\}(\beta_{A,P}-\beta_{A,0}).
\end{align*}
Thus, correctly specifying the treatment mechanism removes the terms involving $(\bar g_P-\bar g_0)$, but does not generally make the numerator remainder zero.
\end{corollary}

\begin{proof}
Set $\bar g_P-\bar g_0=0$ in Lemma \ref{lem:subgroup_remainder_tilde_numerator}.
\end{proof}

Applying Lemma \ref{lem:ratio_remainder_conditional_mean} with $h_P(W)=m_{A,P}(W)$ gives the conditional-mean remainder.
\begin{corollary}\label{cor:subgroup_remainder_tilde_ratio}
\begin{align*}
\tilde R_{\mathcal{M}_{A,w}}(P,P_0)
&=\frac{R_{N_A}(P,P_0)}{p_P}+\frac{p_P-p_0}{p_P}
\left\{\tilde\Psi_{\mathcal{M}_{A,w}}(P)
-\tilde\Psi_{\mathcal{M}_{A,w}}(P_0)\right\}.
\end{align*}
\end{corollary}

We next consider the numerator $N_S(P)=E_P[S\,b_P(W)]$.

\begin{lemma}\label{lem:subgroup_remainder_bias_numerator}
The exact remainder $R_{N_S}(P,P_0):=N_S(P)-N_S(P_0)+P_0D_{N_S,P}$ is given by
\begin{align*}
R_{N_S}(P,P_0)
&=\Biggl\{
E_0\Bigl[S\bigl\{
\Pi_0(0\mid W,0)\phi_S(W,0)
-\Pi_0(0\mid W,1)\phi_S(W,1)
\bigr\}\Bigr]^\top\\
&-
E_P\Bigl[S\bigl\{
\Pi_P(0\mid W,0)\phi_S(W,0)
-\Pi_P(0\mid W,1)\phi_S(W,1)
\bigr\}\Bigr]^\top I_{S,P}^{-1}I_{S,0}
\Biggr\}(\beta_{S,P}-\beta_{S,0})\\
&+E_P\Bigl[S\bigl\{
\Pi_P(0\mid W,0)\phi_S(W,0)
-\Pi_P(0\mid W,1)\phi_S(W,1)
\bigr\}\Bigr]^\top I_{S,P}^{-1}\\
&\times E_0\Bigl[\phi_S(W,A)\Bigl\{
(\Pi_P-\Pi_0)(1\mid W,A)(\bar Q_P-\bar Q_0)(W,A)\\
&-(\Pi_P-\Pi_0)^2(1\mid W,A)
\tau_{S,\beta_{S,P}}(W,A)
\Bigr\}\Bigr]\\
&+E_0\Biggl[\sum_{a=0}^1(2a-1)
\tau_{S,\beta_{S,P}}(W,a)(\Pi_P-\Pi_0)(1\mid W,a)\left\{
\bar\Pi_0(1\mid W)-\bar\Pi_P(1\mid W)
\frac{\bar g_0(a\mid W)}{\bar g_P(a\mid W)}
\right\}\Biggr].
\end{align*}
\end{lemma}

\begin{corollary}\label{cor:subgroup_remainder_bias_known_Pi}
If $\Pi_P=\Pi_0$, only the first term in Lemma \ref{lem:subgroup_remainder_bias_numerator}, involving $(\beta_{S,P}-\beta_{S,0})$, remains. This term need not vanish in the larger model. If, in addition, both distributions satisfy $A\perp(S,W)$ with $P(A=1)=P_0(A=1)=1/2$, and the linear span of $\phi_S$ contains an intercept and the treatment indicator, then $R_{N_S}(P,P_0)=0$.
\end{corollary}

\begin{proof}
When $\Pi_P=\Pi_0$, every term involving $(\Pi_P-\Pi_0)$ vanishes. For the stated randomized-trial case, choose $v$ so that $\phi_S(W,a)^\top v=2-4a$. Direct conditioning on $W$ gives
$$
I_{S,P}v=E_P\Bigl[S\bigl\{\Pi_P(0\mid W,0)\phi_S(W,0)-\Pi_P(0\mid W,1)\phi_S(W,1)\bigr\}\Bigr].
$$
The same identity holds at $P_0$, so the matrix coefficient multiplying $(\beta_{S,P}-\beta_{S,0})$ is zero.
\end{proof}

\begin{remark}\label{rem:balanced_bias_remainder}
More generally, when both distributions satisfy balanced randomization and the span of $\phi_S$ contains an intercept and treatment, the bias numerator remainder simplifies to
\begin{align*}
R_{N_S}(P,P_0)
&=E_0\Bigl[(\bar\Pi_P-\bar\Pi_0)(1\mid W)\bigl\{(\bar Q_P-\bar Q_0)(W,0)-(\bar Q_P-\bar Q_0)(W,1)\bigr\}\Bigr].
\end{align*}
Here $\Pi_P(1\mid W,a)=\bar\Pi_P(1\mid W)$ and $\Pi_0(1\mid W,a)=\bar\Pi_0(1\mid W)$. This simplification requires both distributions, including the fitted distribution, to satisfy the stated randomization restrictions.
\end{remark}

Applying Lemma \ref{lem:ratio_remainder_conditional_mean} with $h_P(W)=b_P(W)$ yields the bias projection remainder.
\begin{corollary}\label{cor:subgroup_remainder_bias_ratio}
\begin{align*}
R^\#_{\mathcal{M}_{S,w}}(P,P_0)=\frac{R_{N_S}(P,P_0)}{p_P}+\frac{p_P-p_0}{p_P}
\left\{\Psi^\#_{\mathcal{M}_{S,w}}(P)
-\Psi^\#_{\mathcal{M}_{S,w}}(P_0)\right\}.
\end{align*}
\end{corollary}

Combining the two components gives the following result.
\begin{theorem}\label{thm:subgroup_exact_remainder}
The exact remainder of the subgroup projection parameter satisfies
\begin{align*}
R_{\mathcal{M}_w}(P,P_0)=\frac{R_{N_A}(P,P_0)-R_{N_S}(P,P_0)}{p_P}+\frac{p_P-p_0}{p_P}
\left\{\Psi_{\mathcal{M}_w}(P)-\Psi_{\mathcal{M}_w}(P_0)\right\},
\end{align*}
where the numerator remainders are given in Lemmas \ref{lem:subgroup_remainder_tilde_numerator} and \ref{lem:subgroup_remainder_bias_numerator}.
\end{theorem}

The numerator remainders include products involving changes in the subgroup-weighted basis moments, information matrices, and projection coefficients. For example,
\begin{align*}
E_0[S\phi_A(W)]^\top-E_P[S\phi_A(W)]^\top I_{A,P}^{-1}I_{A,0}&=\bigl\{E_0[S\phi_A(W)]-E_P[S\phi_A(W)]\bigr\}^\top\\
&+E_P[S\phi_A(W)]^\top I_{A,P}^{-1}(I_{A,P}-I_{A,0}).
\end{align*}
Known randomization removes the terms involving $(\bar g_P-\bar g_0)$, but the remaining product with $(\beta_{A,P}-\beta_{A,0})$ still requires control. Under the additional restrictions in Remark \ref{rem:balanced_bias_remainder}, the bias numerator is a product of errors in subgroup membership and the pooled outcome-regression contrast. Without these restrictions on the fitted distribution, the full expression in Lemma \ref{lem:subgroup_remainder_bias_numerator} must be controlled, including its final expectation.

With fixed-dimensional bases and bounded inverse information matrices, sufficient rate conditions require every displayed term in the numerator remainders to be $o_P(n^{-1/2})$. For data-adaptive bases whose dimensions increase with $n$, corresponding uniform norm and eigenvalue conditions are needed. Finally, if $p_{P_n^\star}-p_0=O_P(n^{-1/2})$, $p_{P_n^\star}$ is bounded away from zero with probability tending to one, and $\Psi_{\mathcal{M}_w}(P_n^\star)-\Psi_{\mathcal{M}_w}(P_0)=o_P(1)$, then the denominator remainder is $o_P(n^{-1/2})$.

\section{Constructing an A-TMLE and practical considerations}\label{sec:atmle_practical}

We now discuss practical aspects of constructing an A-TMLE estimator, including the estimation of nuisance functions, strategies for data-adaptive selection of working models and the recommended loss functions, and computational considerations for efficient implementation.

\subsection{Strategies for estimating nuisance functions and working models}

The construction of an A-TMLE estimator requires estimators for the working models. We use the loss function
$$
\tilde{L}_{\theta_P,\bar{g}_P}(O)(\tau)=\bigg[Y-\theta_P(W)-(A-\bar{g}_P(1\mid W))\tau\bigg]^2
$$
to learn the pooled conditional average treatment effect $\tau_{A,P_0}$, where the loss function is indexed by nuisance functions $\theta_P(W)=E_P(Y\mid W)$ and $\bar{g}_P=P(A=1\mid W)$.
\begin{remark}
The loss function $\tilde{L}$ admits a doubly robust formulation with respect to the treatment mechanism $\bar{g}$ and the pooled outcome regression $\bar{Q}$ when the centering function is constructed as $\theta_{\bar{g},\bar{Q}}(W):=\sum_{a=0}^1 \bar{g}(a\mid W)\bar{Q}(W,a)$, where $\bar{Q}$ is a candidate estimator of $\bar{Q}_P(W,a)=E_P(Y\mid W,A=a)$. Suppose that $\tau_{A,P}\in\mathcal{T}_{A,w}$. Under treatment positivity, $\tau_{A,\beta_{A,P}}=\tau_{A,P}=\arg\min_{\tau\in\mathcal{T}_{A,w}}E_P\left[\tilde{L}_{\theta_{\bar{g},\bar{Q}},\bar{g}}(O)(\tau)\right]$ whenever either $\bar{g}=\bar{g}_P$ or $\bar{Q}=\bar{Q}_P$. Thus, this double robustness concerns $\bar{g}$ and $\bar{Q}$ with $\theta$ constructed from these functions; it does not apply to an independently specified $\theta$.
\end{remark}
Similarly, we use the loss function
$$
L^\#_{\bar{Q}_P,\Pi_P}(O)(\tau)=\bigg[Y-\bar{Q}_P(W,A)-(S-\Pi_P(1\mid W,A))\tau\bigg]^2
$$
to learn the conditional average subgroup effect $\tau_{S,P_0}$, where the loss function is indexed by nuisance functions $\bar{Q}_P(W,A)=E_P(Y\mid W,A)$ and $\Pi_P(1\mid W,A)=P(S=1\mid W,A)$.
\begin{remark}
The loss function $L^\#$ is double robust with respect to misspecification of either $\bar{Q}$ or $\Pi$ in that $\tau_{S,\beta_{S,P}}=\arg\min_{\tau\in\mathcal{T}_{S,w}}E_PL^\#_{\bar{Q},\Pi}(O)(\tau)$ if either $\bar{Q}=\bar{Q}_P$ or $\Pi=\Pi_P$.
\end{remark}
We recommend estimating the nuisance functions $\theta$, $\bar{g}$, $\bar{Q}$, and $\Pi$ using Super Learner, an ensemble approach that allows the incorporation of multiple flexible machine learning algorithms \citep{sl_vdl_2007}. For obtaining working models for conditional average effects $\tau_{A,\beta_A}$ and $\tau_{S,\beta_S}$, we recommend using Highly Adaptive Lasso (HAL) \citep{vdl_generally_2017,benkeser_hal_2016}.

HAL is a nonparametric regression approach that attempts to represent an unknown target function by a very large dictionary of simple basis functions and then estimates the coefficients by empirical risk minimization with an $\ell_1$-type penalty. HAL only assumes that the true function $f_0$ is \textit{càdlàg} (right continuous with left limits) with a bounded variation (as defined in \cite{vdl_generally_2017}), which is arguably a plausible assumption for most functions one may encounter in many applied settings, as it allows discontinuities, threshold effects, nonlinearities, and high-order interactions, while ruling out only functions with arbitrarily oscillatory behavior or infinite variation over the covariate support. In many biomedical, clinical, and trial applications, the relevant regression functions are defined on bounded covariate domains, involve bounded outcomes or probabilities, and are expected to vary in a finite and structured way with baseline covariates. Under this view, the bounded-variation assumption is a plausible regularity condition rather than a restrictive parametric modeling assumption. HAL is motivated by a general representation of multivariate \textit{càdlàg} functions, and it acts as an empirical risk minimizer in that function class. Model complexity is controlled through the $\ell_1$-norm of the basis coefficients, which corresponds to the \textit{sectional variation norm} of the function, a natural quantity for regularization.

More concretely, if $\{\phi_1,\ldots,\phi_J\}$ denotes a dictionary of $J$ HAL basis functions (including an intercept term) and $L(\cdot)$ denotes a loss function, HAL estimates a regression function $f_0$ by solving
$$
f_n=\arg\min_{Q_\beta}\, P_n L(f_\beta)
\quad\text{subject to}\quad
f_\beta=\sum_{j=1}^J \beta_j \phi_j,\qquad \sum_{j=1}^J |\beta_j|\le C_n,
$$
or, equivalently, by a lasso-penalized empirical risk minimization problem. The parameter $C_n$ controls the complexity of the fit and is typically selected by cross-validation. HAL is attractive because it also provides us with a working model defined by the basis functions with non-zero coefficients in the final fit.

As an example, suppose $f_0:[0,1]^3\to\RR$ is a tri-variate function with arguments $X_1$, $X_2$, and $X_3$. Suppose we do not have knowledge on the smoothness of the function. Then a zero-order HAL basis dictionary would consist of indicator basis functions of the form one-way basis functions $I(X_1\geq u)$, $I(X_2\geq u)$, $I(X_3\geq u)$ for $u\in[0,1]$; two-way basis functions $I(X_1\geq u_1,X_2\geq u_2)$, $I(X_1\geq u_1,X_3\geq u_2)$, $I(X_2\geq u_1,X_3\geq u_3)$ for $(u_1,u_2)\in [0,1]^2$; and three-way basis functions $I(X_1\geq u_1,X_2\geq u_2,X_3\geq u_3)$ for $(u_1,u_2,u_3)\in[0,1]^3$. In \cite{benkeser_hal_2016}, the authors proposed to use a set of noninformative knot points implied by the observed data. One could first generate the HAL design matrix with each column a HAL basis function evaluated at the observed data and then run off-the-shelf Lasso optimization software such as the \texttt{R} package \texttt{glmnet}. In practice, it is often recommended to limit the maximum degrees of interactions and number of knot points per subspace of the function for computational feasibility consideration. In Subsection \ref{sec:group_lasso}, we discuss a novel, theoretically motivated approach to perform function subspace screening to alleviate the computational burden of HAL. 

\subsection{TMLE-style targeting of the working model coefficients of $\tilde{\Psi}$ and $\Psi^\#$}\label{subsec:tmle_style_targeting}

In \cite{van2023adaptive}, the authors propose to refit the working-model coefficients unpenalized so that the empirical mean of the corresponding components of the EIC is solved to zero. Here, we propose an alternative, TMLE-style approach for targeting these coefficients.

Let's first consider the parameter $\tilde{\Psi}$. The following lemma presents a closed-form targeting of the working model coefficients $\beta_A$ that is in the spirit of TMLE targeting via a one-dimensional fluctuation submodel.

\begin{lemma}\label{lem:tmle_psi_tilde}
Define $C_{A,n}\coloneq I_{A,n}^{-1}E_n\left[S\phi_A(W)\right]/P_n(S=1)$ as the clever covariate. Let $$
H_{A,n}\coloneq(A-\bar{g}_n(1\mid W))\phi_A(W)C_{A,n},
\qquad
R_{A,n}\coloneq Y-\theta_n(W)-(A-\bar{g}_n(1\mid W))\tau_{A,\beta_{A,n}}(W).
$$ 
Given a one-dimensional linear fluctuation submodel $\epsilon\to\beta_{A,n,\epsilon}\coloneq\beta_{A,n}+\epsilon C_{A,n}$ for $\beta_A$ through its initial estimator $\beta_{A,n}$ in the direction $C_{A,n}$. The closed-form solution $\epsilon^\star$ of $\epsilon$ that solves the score equation $0=P_nH_{A,n}(R_{A,n}-\epsilon H_{A,n})$ is given by
$$
\epsilon^{\star}=\frac{P_n(H_{A,n}R_{A,n})}{P_n(H_{A,n}^2)}.
$$
\end{lemma}
The proof is provided in Appendix \ref{app:tmle_psi_tilde}. Lemma \ref{lem:tmle_psi_tilde} shows that it is not necessary to solve all score equations in the working model using the unpenalized fit. Since our goal is to solve the empirical mean of the EIC presented in Lemma \ref{lem:can_grad_psi_tilde}, we only need to solve that particular one. This preserves all the signals to solve the score equation that really matters from a de-biasing of the target parameter perspective.

Now, similarly, we can derive a similar lemma for $\beta_S$ for the working model for $\Psi^\#$.

\begin{lemma}
Define 
$$
C_{S,n}\coloneq I_{S,n}^{-1}\frac{E_n\Big[S\big\{\Pi_n(0\mid W,0)\phi_S(W,0)-\Pi_n(0\mid W,1)\phi_S(W,1)\big\}\Big]}{P_n(S=1)}
$$
as the clever covariate. Define a one-dimensional linear fluctuation submodel $\gamma\to\beta_{S,n,\gamma}\coloneq\beta_{S,n}+\gamma C_{S,n}$ for $\beta_S$ in the direction $C_{S,n}$.
Let
\begin{align*}
H_{S,n}&\coloneq(S-\Pi_n(1\mid W,A))\phi_S(W,A)C_{S,n}\\
R_{S,n}&\coloneq Y-\bar{Q}_n(W,A)-(S-\Pi_n(1\mid W,A))\tau_{S,\beta_{S,n}}(W,A).
\end{align*}
The closed-form solution $\gamma^\star$ of $\gamma$ that solves the estimating equation $0=P_nH_{S,n}(R_{S,n}-\gamma H_{S,n})$ is given by
$$
\gamma^{\star}=\frac{P_n(H_{S,n}R_{S,n})}{P_n(H_{S,n}^2)}.
$$
\end{lemma}
The proof is provided in Appendix \ref{app:tmle_psi_pound}.

\subsection{Using group Lasso to screen for HAL subspaces}\label{sec:group_lasso}

For zero-order HAL, assuming all the covariates are continuous, the column dimension of the HAL design matrix could go as large as $n\times (2^d-1)$, which is computationally infeasible. Therefore, it is \textit{never} recommended to run the full-blown HAL on reasonably-sized data. There has been some recent work attempting to address the computational aspect of HAL, including \cite{schuler_lassoed_2022,wang_highly_2026,meixide_highly_2026}. Here, we introduce a novel approach that is both practically and theoretically appealing to screen the number of HAL basis functions down to a computationally feasible regime via a group-sparse penalty.

Let $\mathcal{R}_n$ denote a data-adaptively selected index set specifying a collection of HAL basis functions with non-zero coefficients based on the observed sample, and $\mathcal{R}_0$ denote a fixed (oracle) index set corresponding to the oracle submodel that contains the true function of interest. Define $\mathcal{D}(\mathcal{R}_n)$ and $\mathcal{D}(\mathcal{R}_0)$ as the corresponding working models. The A-TMLE framework suggests viewing $\mathcal{D}({\cal R}_n)$ as a data-adaptive working model and requiring it to stabilize relative to a fixed oracle model $\mathcal{D}({\cal R}_0)$ \citep{van2023adaptive}. A convenient regime is the special case $P\bigl(\mathcal{D}({\cal R}_n)\subseteq \mathcal{D}({\cal R}_0)\bigr)\to 1$,
together with the requirement that the sieve within $\mathcal{D}({\cal R}_n)$ approximates $\mathcal{D}({\cal R}_0)$ from below at the required rate (e.g., see Conditions C2 and C4 in \cite{van2023adaptive}). The goal is not necessarily exact recovery of the smallest support of $f_0$, but rather asymptotic stabilization to some fixed oracle set ${\cal R}_0$ such that $f_0\in \mathcal{D}({\cal R}_0)$. The oracle submodel need not be minimal; it may simply be the limiting selected model. A natural sufficient formulation is therefore that ${\cal R}_n$ converges from below to ${\cal R}_0$, in the sense that $P({\cal R}_n\subseteq {\cal R}_0)\to 1$, and the induced sieve within $\mathcal{D}({\cal R}_n)$ becomes sufficiently rich to approximate $\mathcal{D}({\cal R}_0)$ at the required rate. Under this regime, the final estimator lies in the oracle subspace with high probability for large $n$.

HAL basis functions are naturally organized by \emph{subspace}. Each basis function depends on a particular subset of coordinates $s\subset\{1,\ldots,d\}$, where $d$ is the dimension of $f_0$, and a knot point $u$. Therefore, the collection of all basis functions sharing the same coordinate set $s$ can be viewed as representing one additive or interaction component of the target function. Now, this motivates a practical way to construct such a sequence ${\cal R}_n$ using a group-sparsity penalty as a first-stage screen for subspaces. Consider the zero-order HAL case. Suppose we restrict attention to variable interactions of degree at most $m$, giving the collection of coordinate indices
$$
{\cal S}_m=\{s\subset\{1,\dots,d\}: 0<|s|\le m\}.
$$
For each $s\in{\cal S}_m$, let $\beta_s$ denote the block of coefficients corresponding to all zero-order HAL basis functions depending on the same coordinate set $s$ but different knot points. Given a loss function $L(f)(O)$, one may define the group Lasso estimator \citep{yuan_model_2006} given by
$$
\hat\beta=\arg\min_\beta
P_n L(f_\beta)(O)+\lambda \sum_{s\in{\cal S}_m} w_s \|\beta_s\|_2,
$$
with weights $w_s$ chosen, for example, as $w_s=\sqrt{p_s}$, where $p_s$ is the number of basis functions in group $s$. The role of this penalty is to induce sparsity at the level of subspaces rather than individual basis functions. That is, for many interaction sets $s$, the entire coefficient block $\beta_s$ is set to zero. In this sense, the group Lasso acts as a screening rule for additive components and interactions, producing a relatively small collection of candidate subspaces that appear to matter.

This screening step can then be combined with a second-stage HAL-MLE. Specifically, one may first run the group Lasso with a small number of knot points for each subspace, over a range of $\lambda$ values, to generate a ranked family of candidate supports $\{{\cal R}_{0,\ell}:\ell=1,2,\dots\}$. Then, for each candidate support, one fits a HAL-MLE restricted to the corresponding subspace $\mathcal{D}({\cal R}_{0,\ell})$, using a much denser knot point grid. Cross-validation may then be used over the candidate support index $\ell$ to select the final working model. Conceptually, this mirrors the setup of cross-validation over a growing sequence of working models described in \cite{van2023adaptive}. Specifically, the screening step selects a candidate oracle submodel, and the final within-subspace HAL fit may be undersmoothed, possibly, only relative to that selected submodel, rather than relative to the full function space.

We present the pseudo-code for the implementation of A-TMLE in Algorithm \ref{alg:atmle}.

\begin{breakablealgorithm}
\footnotesize
\caption{Pseudo-code for A-TMLE}
\label{alg:atmle}
\begin{algorithmic}[1]
\STATE Obtain estimators $\bar{g}_n$ of $\bar{g}_0$ and $\theta_n$ of $\theta_0$;
\STATE Fit a HAL with basis functions $\{\phi_{A,j}(W):j\}$:
$$
\beta_{A,n}=\arg\min_{\beta_A,\norm{\beta_A}_1\leq M_{A,n}}
P_n\Bigl[
\bigl\{
Y-\theta_n(W)-(A-\bar{g}_n(1\mid W))\tau_{A,\beta_A}(W)
\bigr\}^2
\Bigr];
$$
\STATE Define $\mathcal{R}_{A,n}$ as the index set of the non-zero entries of $\beta_{A,n}$, and restrict $\tau_{A,\beta_A}$ to the selected working model generated by $\{\phi_{A,j}(W):j\in\mathcal{R}_{A,n}\}$;
\STATE Define
$$
C_{A,n}\coloneq I_{A,n}^{-1}\frac{P_n\{S\phi_A(W)\}}{P_n(S=1)},
\qquad
H_{A,n}\coloneq (A-\bar{g}_n(1\mid W))\phi_A(W)C_{A,n},
$$
and
$$
R_{A,n}\coloneq Y-\theta_n(W)-(A-\bar{g}_n(1\mid W))\tau_{A,\beta_{A,n}}(W).
$$
Compute
$$
\epsilon_{A,n}^{\star}\coloneq \frac{P_n(H_{A,n}R_{A,n})}{P_n(H_{A,n}^2)},
$$
and update
$$
\beta_{A,n}^{\star}\coloneq \beta_{A,n}+\epsilon_{A,n}^{\star}C_{A,n}.
$$
\STATE Obtain estimators $\Pi_n^{(0)}$ of $\Pi_0$ and $\bar{Q}_n$ of $\bar{Q}_0$;
\STATE Fit a HAL with basis functions $\{\phi_{S,j}(W,A):j\}$:
$$
\beta_{S,n}^{(0)}=\arg\min_{\beta_S,\norm{\beta_S}_1\leq M_{S,n}}
P_n\Bigl[
\bigl\{
Y-\bar{Q}_n(W,A)-(S-\Pi_n^{(0)}(1\mid W,A))\tau_{S,\beta_S}(W,A)
\bigr\}^2
\Bigr];
$$
\STATE Define $\mathcal{R}_{S,n}$ as the index set of the non-zero entries of $\beta_{S,n}^{(0)}$, and restrict $\tau_{S,\beta_S}$ to the selected working model generated by $\{\phi_{S,j}(W,A):j\in\mathcal{R}_{S,n}\}$;
\STATE Define the initial estimate 
$$
P_n^{(0)}\coloneq(\bar{g}_n,\theta_n,\beta_{A,n}^{\star},\Pi_n^{(0)},\bar{Q}_n,\beta_{S,n}^{(0)}),
$$ 
and the convergence threshold 
$$
s_n\coloneq \sigma_n^{(0)}/(\sqrt{n}\log n),$$
where $\sigma_n^{(0)}$ is the estimated standard deviation of the efficient influence curve of $\Psi_{\mathcal{M}_w}$ at $P_n^{(0)}$. Set iteration counter $k\coloneq 0$;
\WHILE{$|P_nD_{\mathcal{M}_w,P_n^{(k)}}|>s_n$}
\STATE Target $\Pi$ by fitting a univariate logistic regression with weights $\bar{\Pi}_n(1\mid W)/P_n(S=1)$:
$$
\text{logit}\,\Pi_n^{(k+1)}(1\mid W,A)
=
\text{logit}\,\Pi_n^{(k)}(1\mid W,A)
+\varepsilon_n^{(k)}
\left\{
\frac{A}{\bar{g}_n(1\mid W)}\tau_{S,\beta_{S,n}^{(k)}}(W,1)-\frac{1-A}{\bar{g}_n(0\mid W)}\tau_{S,\beta_{S,n}^{(k)}}(W,0)
\right\},
$$
where $\varepsilon_n^{(k)}$ is the fitted coefficient;
\STATE Define
$$
C_{S,n}^{(k+1)}
\coloneq
 I_{S,n}^{(k+1)-1}
\frac{
P_n\Bigl[
S\Bigl\{
\Pi_n^{(k+1)}(0\mid W,0)\phi_S(W,0)
-
\Pi_n^{(k+1)}(0\mid W,1)\phi_S(W,1)
\Bigr\}
\Bigr]
}{
P_n(S=1)
},
$$
$$
H_{S,n}^{(k+1)}
\coloneq
(S-\Pi_n^{(k+1)}(1\mid W,A))\phi_S(W,A)C_{S,n}^{(k+1)},
$$
and
$$
R_{S,n}^{(k+1)}
\coloneq
Y-\bar{Q}_n(W,A)-(S-\Pi_n^{(k+1)}(1\mid W,A))\tau_{S,\beta_{S,n}^{(k)}}(W,A).
$$
Compute
$$
\gamma_n^{(k+1)\star}
\coloneq
\frac{P_n(H_{S,n}^{(k+1)}R_{S,n}^{(k+1)})}{P_n((H_{S,n}^{(k+1)})^2)},
$$
and update
$$
\beta_{S,n}^{(k+1)}
\coloneq
\beta_{S,n}^{(k)}+\gamma_n^{(k+1)\star}C_{S,n}^{(k+1)}.
$$
\STATE Define
$$
P_n^{(k+1)}
\coloneq
(\bar{g}_n,\theta_n,\beta_{A,n}^{\star},\Pi_n^{(k+1)},\bar{Q}_n,\beta_{S,n}^{(k+1)}),
$$
and update
$$
s_n\coloneq \sigma_n^{(k+1)}/(\sqrt{n}\log n),
$$
where $\sigma_n^{(k+1)}$ is the estimated standard deviation of the efficient influence curve of $\Psi_{\mathcal{M}_w}$ at $P_n^{(k+1)}$;
\STATE Increase iteration counter $k$ by 1;
\ENDWHILE
\STATE Define $K\coloneq k$, and compute the plug-in estimator
$$
\Psi_{\mathcal{M}_w}(P_n^{(K)})
=
\tilde{\Psi}_{\mathcal{M}_{A,w}}(\beta_{A,n}^{\star})
-
\Psi_{\mathcal{M}_{S,w}}^\#(\Pi_n^{(K)},\beta_{S,n}^{(K)}),
$$
where
$$
\tilde{\Psi}_{\mathcal{M}_{A,w}}(\beta_{A,n}^{\star})
=
\frac{P_n\bigl[S\,\tau_{A,\beta_{A,n}^{\star}}(W)\bigr]}{P_n(S=1)},
$$
and
$$
\Psi_{\mathcal{M}_{S,w}}^\#(\Pi_n^{(K)},\beta_{S,n}^{(K)})
=
\frac{
P_n\Bigl[
S\Bigl\{
\Pi_n^{(K)}(0\mid W,0)\tau_{S,\beta_{S,n}^{(K)}}(W,0)
-
\Pi_n^{(K)}(0\mid W,1)\tau_{S,\beta_{S,n}^{(K)}}(W,1)
\Bigr\}
\Bigr]
}{
P_n(S=1)
}.
$$
\STATE Compute the Wald-type 95\% confidence interval
$$
\Psi_{\mathcal{M}_w}(P_n^{(K)})
\pm
1.96\sqrt{P_n\bigl(D_{\mathcal{M}_w,P_n^{(K)}}^2\bigr)/n}.
$$
\end{algorithmic}
\end{breakablealgorithm}

\section{Simulation study}\label{sec:simulation}

The goal of this simulation study is to compare estimators for subgroup analysis that are capable of borrowing information from the remainder of the trial while preserving nominal confidence interval coverage. As benchmarks, we include three estimators fit only using subgroup data: an unadjusted difference-in-means estimator, a subgroup-only augmented inverse probability weighted (AIPW) estimator \citep{robins_estimation_1994,robins_analysis_1995,vdl_unified_2003, tsiatis_semiparametric_2006}, and a subgroup-only TMLE (as described in Remark \ref{rem:subgroup_tmle}). We compare these subgroup-only estimators with TMLE-PR and A-TMLE. Considering our LEADER real-data application, we designed our simulation study to closely mimic the data-generating distribution that likely generated the real LEADER trial data. While LEADER data are not publicly available, to make the simulation reproducible, we generate baseline covariates based on summary statistics and dependence structures of the observed LEADER trial data.

\subsection{Synthetic data-generating processes based on LEADER}\label{subsec:dgps}

For demonstration purposes, we consider the Asian subgroup as the subgroup of interest here. We generated baseline covariates to resemble the LEADER trial population using a Gaussian copula model with empirically estimated dependence structure and marginal distributions matched to reported summary statistics of patient baseline characteristics of LEADER. Specifically, let $Z=(Z_1,\dots,Z_d)$ denote a multivariate normal random vector with mean zero and correlation matrix $R$. We transformed each component to a uniform variable $U_j=\Phi(Z_j)$, where $\Phi(\cdot)$ is the standard normal cumulative distribution function (CDF). These uniforms were then mapped to observed baseline covariates. For continuous baseline covariates, we used a piecewise-linear inverse CDF constructed from the reported minimum, first quartile, median, third quartile, and maximum. Thus, if $W_j$ denotes a continuous covariate and $q_j(p)$ its approximate quantile function, then $W_j=\tilde{q}_j(U_j)$, where $\tilde{q}_j$ is the linear interpolation through the five reported quantile points. This process was done for age, diabetes duration, estimated glomerular filtration rate (eGFR), body mass index (BMI), glycated hemoglobin (HbA1c), low-density lipoprotein (LDL) cholesterol, and high-density lipoprotein (HDL) cholesterol. For binary baseline covariates, we used Bernoulli draws induced by the latent uniforms, i.e., $W_j=I(U_j<p_j)$, where $p_j$ was chosen to match the reported marginal prevalence. This process was done for sex, insulin-naive status (INSNVFL), antihypertensive treatment, and prior cardiovascular disease status (PRCVFL).

The matrix $R$ was constructed from an empirical covariance matrix of baseline covariates computed from the individual-level LEADER trial data and supplied as an aggregate summary. Specifically, let $\widehat{\Sigma}$ denote the covariance matrix for age, sex, diabetes duration, insulin-naive status, antihypertensive treatment, eGFR, BMI, HbA1c, LDL cholesterol, HDL cholesterol, and prior cardiovascular disease status. We standardized this matrix as $R_{jk}=\widehat{\Sigma}_{jk}/\sqrt{\widehat{\Sigma}_{jj}\widehat{\Sigma}_{kk}}$ for $j,k=1,\ldots,11$, and used the resulting matrix as the correlation matrix of the latent Gaussian vector $Z$. The same matrix was used across simulation scenarios and replications. The supplied matrix was positive definite, so no positive-definiteness adjustment was required. Treatment assignment was excluded from this matrix and generated independently according to the trial randomization probability.

Throughout, we let $W$ denote the baseline covariates other than race, and let $S\in\{0,1\}$ denote the subgroup indicator for Asian race. Asian race was generated separately from the Gaussian copula model using a logistic model depending on baseline BMI, so that the simulated Asian subgroup was both approximately the same size as in LEADER and leaner on average. Let $z_{\text{BMI}}$ be standardized BMI. We generated $S\sim\text{Bernoulli}\{\text{expit}(\alpha_0+\alpha_1 z_{\text{BMI}})\}$, where $\alpha_0$ was chosen so that the marginal prevalence of Asian participants matched the LEADER race proportion, $920/9340\approx 0.099$, and $\alpha_1=-1.131383<0$ made Asian race more common among individuals with lower BMI.

The treatment was randomized independently of baseline covariates, that is, $A\sim\text{Bernoulli}(1/2)$. We designed the following three scenarios for the outcome process. These scenarios correspond to three broad classes of event-generating mechanisms commonly encountered in time-to-event settings, including (1) simple proportional-hazards (PH), (2) complex PH, and (3) complex non-PH, thereby allowing us to examine estimator performance across a range of plausible data-generating regimes. These scenarios were intentionally constructed so that the subgroup-specific treatment effect among Asian patients differs substantially from the pooled treatment effect evaluated over the Asian covariate distribution. In all three scenarios, this discrepancy is induced by explicit $A\times S$ interaction terms. Consequently, these settings represent an intentionally unfavorable setting for estimators that rely on the pooled estimand $\tilde{\Psi}$ naively when the scientific target is subgroup-specific, $\Psi$, because the resulting bias parameter $\Psi^\#$ is large in magnitude by design. We emphasize that this setup should be viewed as a stress-test rather than as a literal representation of the LEADER trial. Indeed, in our real-data analysis of LEADER presented in the following section, the estimated bias parameter is quite small. Nevertheless, from a simulation standpoint, it is useful to first examine this adverse regime in order to assess estimator behavior when the true bias estimand is non-negligible.

\paragraph{Scenario 1: Simple proportional-hazards.}

In this scenario, we generated time to event from a Weibull proportional-hazards (PH) model. Conditional on treatment $A$, subgroup status $S$, and baseline covariates $W$, the hazard function took the form $\lambda_1(t\mid S,A,W)=\eta\nu t^{\nu-1}\exp\bigl(\beta_W^\top W+\beta_AA+\beta_{AS}AS\bigr),$ where $\nu=1.5$ is the Weibull shape parameter, $\eta$ is the scale parameter, $\beta_A=\log(0.95)$, and $\beta_{AS}=\log(0.55)-\log(0.95)$. Thus, the treatment hazard ratio is $0.95$ for non-Asian participants ($S=0$) and $0.55$ for Asian participants ($S=1$). The Weibull scale parameter $\eta$ is calibrated by Monte Carlo so that, under 1:1 randomization, the overall pooled event risk by 3.8 years matches the LEADER target event proportion, $1302/9340$. Event times were simulated by inverse transform sampling:
$$
T=\left\{\frac{-\log(U)}{\eta\exp\bigl(\beta_W^\top W+\beta_AA+\beta_{AS}AS\bigr)}\right\}^{1/\nu},
$$
where $U\sim\text{Uniform}(0,1)$.

To simplify the design, we assumed no dropout or loss to follow-up. LEADER actually had very high retention and nearly complete endpoint ascertainment. In published follow-up summaries, 96.8\% of randomized participants either completed the final visit, died, or experienced a primary outcome event, and vital status was known for 99.7\% of participants \citep{marso_liraglutide_2016}. In our simulation study, all participants were administratively censored at a common fixed follow-up time $\tau=3.8$ years, chosen to reflect the median follow-up in LEADER. Thus, the observed follow-up time and event indicator were $\tilde{T}=\min(T,\tau)$ and $\Delta=I(T\leq\tau)$. The binary observed outcome used in the simulation was therefore $Y=I(T\leq\tau)=\Delta$, that is, whether the event occurred by 3.8 years. Under this data-generating process,
$$
E(Y\mid S=s,A=a,W)=1-\exp\left\{-\eta\tau^\nu\exp\bigl(\beta_W^\top W+a(\beta_A+\beta_{AS}s)\bigr)\right\}.
$$

\paragraph{Scenario 2: Complex proportional-hazards.}

In this scenario, the hazard still satisfies the PH assumption, but it involves a LEADER-style linear prognostic score together with additional non-linear terms. First, we define the centered/scaled covariates
\begin{gather*}
\text{age}_c=\frac{\text{age}-64}{10},\qquad
\text{dur}_c=\frac{\text{diabetes duration}-11.5}{10},\qquad
\text{egfr}_c=\frac{\text{eGFR}-84.3}{20},\\
\text{bmi}_c=\frac{\text{BMI}-31.7}{10},\qquad
\text{hba1c}_c=\frac{\text{HbA1c}-8.3}{2},\qquad
\text{ldl}_c=\text{LDL}-2.20.
\end{gather*}
Let $\ell(W)$ denote the linear prognostic score used in Scenario 1, and define a non-linear prognostic function:
\begin{align*}
h(W)&=0.08\,\text{age}_c^2+0.06\,\text{dur}_c+0.08\sin(\pi\text{hba1c}_c)+0.06\,\text{ldl}_c\text{egfr}_c\\
&\quad+0.05\,\text{bmi}_c^2+0.08\,\text{PRCVFL}+0.06\,\text{INSNVFL}.
\end{align*}
We then define the log-relative hazard as $f_2(S,A,W)=\ell(W)+2.5\{h(W)+0.04S\}+\beta_AA+\beta_{AS}AS$, with $\beta_A=\log(0.95)$ and $\beta_{AS}=\log(0.55)-\log(0.95)$. Therefore, $\lambda_2(t\mid S,A,W)=\eta\nu t^{\nu-1}\exp\{f_2(S,A,W)\}$. This is still a PH model because the treatment effect enters multiplicatively through a time-constant term. Under this scenario,
$$
E(Y\mid S=s,A=a,W)=1-\exp\left\{-\eta\tau^\nu\exp\bigl(\ell(W)+2.5\{h(W)+0.04s\}+a(\beta_A+\beta_{AS}s)\bigr)\right\}.
$$

\paragraph{Scenario 3: Complex non-proportional-hazards.}

We now design a scenario in which the hazard is piecewise in time, so that the PH assumption is violated. We use the same linear score $\ell(W)$ and non-linear prognostic function $h(W)$ as in Scenario 2, and let the change-point be $t_0=1.5$ years. Define $\beta_A^{(1)}=\log(0.98)$ and $\beta_A^{(2)}=\log(0.90)$ for the early and late treatment effects among non-Asian participants, and $\beta_{AS}^{(1)}=\log(0.70)-\log(0.98)$ and $\beta_{AS}^{(2)}=\log(0.45)-\log(0.90)$ for the additional subgroup-treatment interaction in the early and late periods. Then the hazard is
$$
\lambda_3(t\mid S,A,W)=
\begin{cases}
\eta\nu t^{\nu-1}\exp\bigl\{\ell(W)+2.5\{h(W)+0.04S\}+A(\beta_A^{(1)}+\beta_{AS}^{(1)}S)\bigr\},&t<t_0,\\[6pt]
\eta\nu t^{\nu-1}\exp\bigl\{\ell(W)+2.5\{h(W)+0.04S\}+A(\beta_A^{(2)}+\beta_{AS}^{(2)}S)\bigr\},&t\ge t_0.
\end{cases}
$$
For $\tau>t_0$, the cumulative hazard is
\begin{align*}
\Lambda_{3}(\tau\mid S=s,A=a,W)&=\eta\Big[
t_0^\nu\exp\bigl\{\ell(W)+2.5\{h(W)+0.04s\}+a(\beta_A^{(1)}+\beta_{AS}^{(1)}s)\bigr\}\\
&+(\tau^\nu-t_0^\nu)\exp\bigl\{\ell(W)+2.5\{h(W)+0.04s\}+a(\beta_A^{(2)}+\beta_{AS}^{(2)}s)\bigr\}\Big]
\end{align*}
and therefore
$$
E(Y\mid S=s,A=a,W)=1-\exp\left\{-\Lambda_{3}(\tau\mid S=s,A=a,W)\right\}.
$$

\paragraph{}
For all three scenarios, the true causal estimands were evaluated by Monte Carlo integration under the known data-generating processes. Specifically, we generated a large sample of baseline covariates $(W,S)$ from the same baseline distribution used in the simulation. For each scenario $j\in\{1,2,3\}$, subgroup value $s\in\{0,1\}$, treatment level $a\in\{0,1\}$, and time $t$, we computed the conditional event risk $F_{j}(t\mid S=s,A=a,W)=P(Y=1\mid S=s,A=a,W)$. We then computed the truth as:
$$
\psi_{j}(t)=E\bigl[F_{j}(t\mid S=1,A=1,W)-F_{j}(t\mid S=1,A=0,W)\mid S=1\bigr],
$$
the subgroup-specific treatment effect averaged over the Asian covariate distribution.

\paragraph{Bias function under each scenario.}

Now, we compute the bias function, i.e., $\tau_{S,0}(W,A)=E_0(Y\mid S=1,A,W)-E_0(Y\mid S=0,A,W)$ under each scenario. This quantity is non-zero in all scenarios because the outcome model contains explicit $A\times S$ interaction terms, so the effect of treatment differs between subgroup and non-subgroup patients even after conditioning on $W$. Under Scenario 1, recall that
$$
E(Y\mid S=s,A=a,W)=1-\exp\left\{-\eta\tau^\nu\exp\bigl(\beta_W^\top W+a(\beta_A+\beta_{AS}s)\bigr)\right\}.
$$
Therefore,
$$
\tau_{S,0}(A,W)=\exp\left\{-\eta\tau^\nu\exp\bigl(\beta_W^\top W+A\beta_A\bigr)\right\}
-\exp\left\{-\eta\tau^\nu\exp\bigl(\beta_W^\top W+A(\beta_A+\beta_{AS})\bigr)\right\}.
$$
In particular, when $A=0$, this bias function is zero, whereas when $A=1$ it is non-zero because the treatment hazard differs between $S=1$ and $S=0$ through the interaction coefficient $\beta_{AS}$. Under Scenario 2, we have
$$
E(Y\mid S=s,A=a,W)=1-\exp\left\{-\eta\tau^\nu\exp\bigl(\ell(W)+2.5\{h(W)+0.04s\}+a(\beta_A+\beta_{AS}s)\bigr)\right\},
$$
so the bias function is
$$
\tau_{S,0}(W,A)=\exp\left\{-\eta\tau^\nu\exp\bigl(\ell(W)+2.5h(W)+A\beta_A\bigr)\right\}
-\exp\left\{-\eta\tau^\nu\exp\bigl(\ell(W)+2.5\{h(W)+0.04\}+A(\beta_A+\beta_{AS})\bigr)\right\}.
$$
Unlike Scenario 1, this bias function is already non-zero at $A=0$ because $S$ also enters the prognostic part of the hazard through the term $2.5\times0.04S$. When $A=1$, the bias additionally reflects the subgroup-treatment interaction. Under Scenario 3, recall that we have
$$
E(Y\mid S=s,A=a,W)=1-\exp\left\{-\Lambda_{3}(\tau\mid S=s,A=a,W)\right\},
$$
and hence
$$
\tau_{S,0}(A,W)=\exp\left\{-\Lambda_{3}(\tau\mid S=0,A,W)\right\}
-\exp\left\{-\Lambda_{3}(\tau\mid S=1,A,W)\right\}.
$$
Again, this bias function is non-zero both because $S$ is prognostic and because the treatment effect is allowed to differ by subgroup in both the early and late hazard components.

\subsection{Candidate estimators compared in the simulation study}\label{subsec:candidate_estimators}

We compare the performance of five candidate estimators. First, as a reference, we consider an unadjusted estimator on the subgroup-only data. Since conditional on the subgroup, treatment is still randomized, the subgroup-only unadjusted estimator is unbiased and asymptotically normal, but not efficient. Second, we consider a subgroup-only AIPW estimator and a subgroup-only TMLE estimator. Due to covariate adjustment in these two methods, if the baseline covariates are prognostic, we expect them to be more efficient than the unadjusted estimator \citep{schuler_increasing_2022,emilie_within_2026}. Third, we consider TMLE-PR, described in Section \ref{sec:tmle_pr}, which improves upon the subgroup-only TMLE by fitting the initial outcome regression on the pooled RCT data, allowing non-subgroup observations to help fit the subgroup-specific outcome regression. Finally, we consider the subgroup A-TMLE discussed in Section \ref{sec:atmle}. To ensure a fair comparison, we use the same Super Learner library for all relevant nuisance function estimation. The details on the configurations of the candidate estimators are available in Appendix \ref{app:implementation_details}.

\subsection{Simulation results}

Results are presented in Table \ref{tab:leader_largebias_n2000_asian}. Across all three scenarios, the unadjusted estimator has the largest empirical standard error and MSE, as expected when baseline covariates are prognostic. The subgroup-only AIPW and subgroup-only TMLE estimators show nearly identical performance and improve modestly over the unadjusted estimator through covariate adjustment. Both TMLE-PR and A-TMLE further reduce empirical standard error and MSE relative to the subgroup-only estimators. Between the two proposed estimators, A-TMLE is more centered and maintains approximately nominal empirical coverage, whereas TMLE-PR achieves the smallest empirical standard error and MSE but has conservative confidence interval coverage. Overall, the results demonstrate clear and consistent efficiency gains for both TMLE-PR and A-TMLE compared to the unadjusted and covariate-adjusted subgroup-only estimators across all three scenarios.

TMLE-PR exhibited conservative confidence interval coverage. One possible contributor is estimation of the treatment mechanism, despite its being known by randomization. Under suitable conditions, estimating the propensity score can reduce sampling variance when the outcome regression is misspecified, while the usual plug-in influence-function variance does not capture this reduction \citep{moore_covariate_2009}. This provides a possible explanation for the observed over-coverage.

Figure \ref{fig:bias_correction} is useful for visualizing the bias-correction of A-TMLE in action. We can clearly see that the A-TMLE is not just going after the pooled estimand naively. It is carrying out a meaningful bias correction. This is exactly the behavior one would hope to see given the design of the simulation, where the subgroup-specific treatment effect differs substantially from the pooled treatment effect over the subgroup covariate distribution. The visible gap between the $\tilde{\Psi}$ estimates and the final subgroup-specific estimates confirms that the estimated bias parameter is practically important in these scenarios. We emphasize that a feature of A-TMLE is its ability to achieve efficiency gains even when the bias parameter is large in magnitude. In contrast, many ``test-then-pool'' data integration approaches derive their gains primarily under small bias. In the current setting, where the bias is large relative to the standard error, such methods would tend to reject pooling altogether and revert to subgroup-only estimation. Consequently, we would not expect them to outperform subgroup-only AIPW or TMLE under this regime.

\begin{table}[htbp]
\centering
\resizebox{16cm}{!}{
\begin{tabular}{crccccc}
\toprule
Scenario & Estimator & $|\text{Bias}|$ ($\times 10^{-2}$) & SE ($\times 10^{-2}$) & MSE ($\times 10^{-2}$) & Coverage & Oracle Coverage \\
\midrule
\multirow{5}{*}{Scenario 1} & Unadjusted & 0.088 & 4.593 & 0.211 & 0.94 & 0.95
\\
 & AIPW & 0.085 & 4.089 & 0.167 & 0.96 & 0.95
\\
 & TMLE & 0.105 & 4.070 & 0.165 & 0.96 & 0.96
\\
 & TMLE-PR & 0.425 & 3.544 & 0.127 & 0.98 & 0.96
\\
 & A-TMLE & 0.035 & 3.757 & 0.141 & 0.95 & 0.94
\\
\midrule
\multirow{5}{*}{Scenario 2} & Unadjusted & 0.069 & 4.145 & 0.172 & 0.96 & 0.96
\\
 & AIPW & 0.038 & 3.998 & 0.160 & 0.97 & 0.97
\\
 & TMLE & 0.087 & 3.969 & 0.157 & 0.97 & 0.96
\\
 & TMLE-PR & 0.454 & 3.392 & 0.117 & 0.99 & 0.96
\\
 & A-TMLE & 0.086 & 3.525 & 0.124 & 0.96 & 0.95
\\
\midrule
\multirow{5}{*}{Scenario 3} & Unadjusted & 0.099 & 4.127 & 0.170 & 0.95 & 0.96
\\
 & AIPW & 0.101 & 4.026 & 0.162 & 0.97 & 0.95
\\
 & TMLE & 0.153 & 4.059 & 0.165 & 0.95 & 0.96
\\
 & TMLE-PR & 0.472 & 3.351 & 0.114 & 0.98 & 0.96
\\
 & A-TMLE & 0.088 & 3.451 & 0.119 & 0.97 & 0.96
\\
\bottomrule
\end{tabular}
}
\caption{Performance of the candidate estimators under the scenarios described in subsection \ref{subsec:dgps} at sample size $n=2,000$ across 500 Monte Carlo replications for the Asian-subgroup target estimand. Reported quantities are the absolute bias, empirical standard error (SE), mean squared error (MSE), empirical 95\% confidence interval coverage using the estimated standard error, and oracle coverage using the Monte Carlo standard deviation in place of the estimated standard error. Bias, SE, and MSE are reported on the scale of the target parameter and multiplied by $10^{-2}$.}
\label{tab:leader_largebias_n2000_asian}
\end{table}

\begin{figure}[htbp]
\centering
\includegraphics[width=0.8\linewidth]{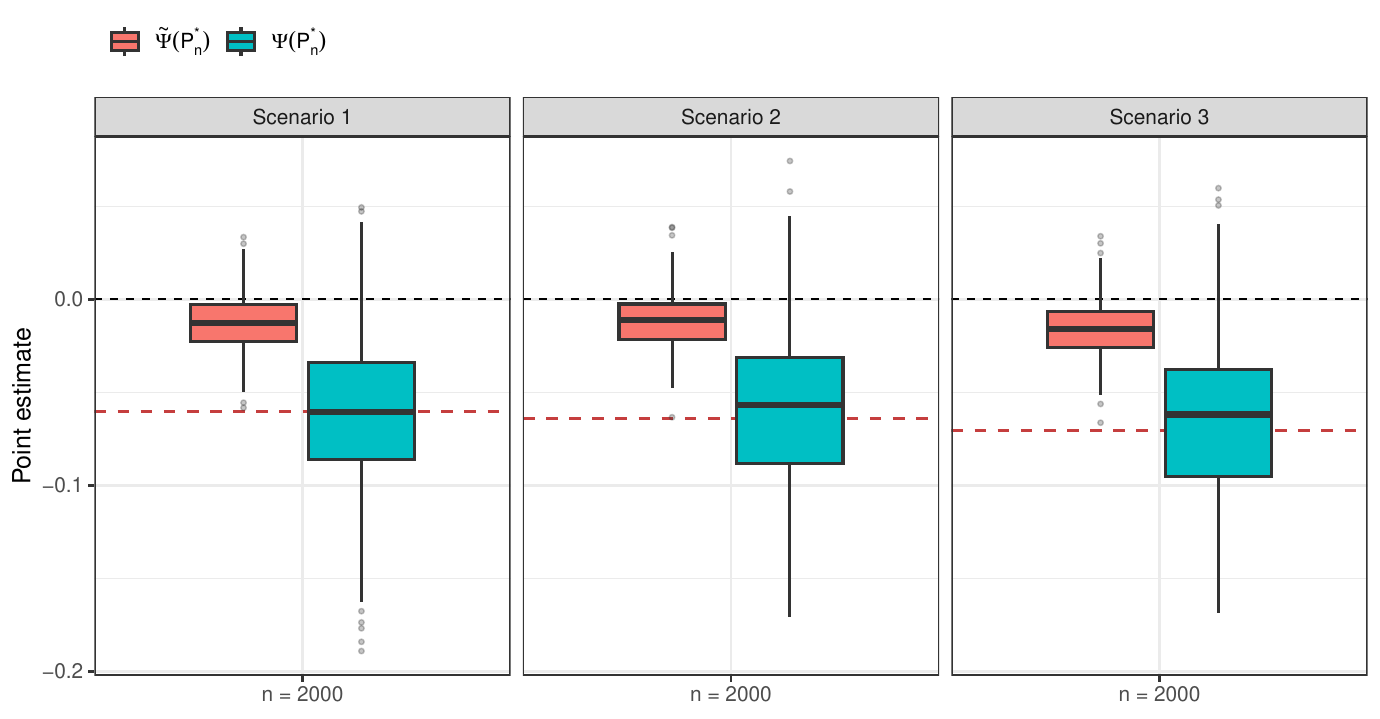}
\caption{Distributions of the pooled projection estimate $\tilde{\Psi}$ and the final bias-corrected subgroup estimate $\Psi=\tilde{\Psi}-\Psi^\#$ across the simulation scenarios at sample size $n=2,000$ over 500 Monte Carlo runs of A-TMLE. The red dashed line indicates the true subgroup-specific target estimand $\Psi$. The separation between the pooled and bias-corrected estimates reflects the magnitude of the estimated bias parameter $\Psi^\#$ and demonstrates that A-TMLE is performing a non-trivial bias correction.}
\label{fig:bias_correction}
\end{figure}

\subsection{Simulations where baseline covariates are strongly predictive of subgroup status}

The previous simulation setting was designed to stress test the performance of the estimators under an unfavorable setting where the subgroup-specific effect was very different from the overall effect. However, in our real-data application, we have seen that the estimate of the bias parameter $\Psi^\#$ is actually quite small, and there is also practical concern that the baseline covariates might be strongly predictive of whether an individual belongs to a subgroup. Inspired by these observations, we conducted additional simulations that kept the same three classes of event-time structures but used the original small-bias treatment-effect specifications, removing the large explicit subgroup-treatment interaction and making subgroup membership strongly predictable from baseline covariates. Specifically, we generated $S$ from a logistic model with linear predictor $\alpha_0-5z_{\mathrm{BMI}}+3z_{\mathrm{age}}+2z_{\mathrm{eGFR}}+1.5\,\mathrm{SEX}+1.5\,\mathrm{INSNVFL}-2\,\mathrm{PRCVFL}$, where $\alpha_0$ was calibrated to preserve the Asian subgroup prevalence $920/9340$.
\begin{table}[htbp]
\centering
\resizebox{16cm}{!}{
\begin{tabular}{crccccc}
\toprule
Scenario & Estimator & $|\text{Bias}|$ ($\times 10^{-2}$) & SE ($\times 10^{-2}$) & MSE ($\times 10^{-2}$) & Coverage & Oracle Coverage \\
\midrule
\multirow{5}{*}{Scenario 1} & Unadjusted & 0.006 & 4.931 & 0.243 & 0.94 & 0.94
\\
 & AIPW & 0.082 & 4.387 & 0.192 & 0.95 & 0.95
\\
 & TMLE & 0.076 & 4.449 & 0.198 & 0.95 & 0.95
\\
 & TMLE-PR & 0.207 & 3.878 & 0.151 & 0.98 & 0.94
\\
 & A-TMLE & 0.291 & 3.191 & 0.102 & 0.95 & 0.95
\\
\midrule
\multirow{5}{*}{Scenario 2} & Unadjusted & 0.005 & 5.215 & 0.271 & 0.95 & 0.95
\\
 & AIPW & 0.023 & 4.819 & 0.232 & 0.95 & 0.95
\\
 & TMLE & 0.017 & 4.916 & 0.241 & 0.94 & 0.95
\\
 & TMLE-PR & 0.038 & 4.171 & 0.174 & 0.98 & 0.94
\\
 & A-TMLE & 0.163 & 3.152 & 0.099 & 0.95 & 0.95
\\
\midrule
\multirow{5}{*}{Scenario 3} & Unadjusted & 0.002 & 5.243 & 0.274 & 0.95 & 0.94
\\
 & AIPW & 0.011 & 4.780 & 0.228 & 0.95 & 0.95
\\
 & TMLE & 0.019 & 4.874 & 0.237 & 0.94 & 0.95
\\
 & TMLE-PR & 0.052 & 4.216 & 0.177 & 0.98 & 0.93
\\
 & A-TMLE & 0.134 & 3.251 & 0.106 & 0.94 & 0.95
\\
\bottomrule
\end{tabular}
}
\caption{Performance of the candidate estimators under the setting where baseline covariates strongly predict subgroup membership at sample size $n=2,000$ across 500 Monte Carlo replications for the Asian-subgroup target estimand. Reported quantities are the absolute bias, empirical standard error (SE), mean squared error (MSE), empirical 95\% confidence interval coverage using the estimated standard error, and oracle coverage using the Monte Carlo standard deviation in place of the estimated standard error. Bias, SE, and MSE are reported on the scale of the target parameter and multiplied by $10^{-2}$.}
\label{tab:leader_highSpred_n2000_asian}
\end{table}

Results are presented in Table \ref{tab:leader_highSpred_n2000_asian}. As before, the absolute bias is small relative to the standard error for all estimators across all three scenarios, so differences in performance are mainly driven by variance. The unadjusted estimator has the largest MSE in every scenario. Subgroup-only AIPW and subgroup-only TMLE improve over the unadjusted estimator, while TMLE-PR further reduces empirical standard error and MSE. However, TMLE-PR remains noticeably less efficient than A-TMLE. A-TMLE delivers the largest efficiency gains, reducing the MSE by more than half relative to the unadjusted estimator while maintaining approximately nominal empirical coverage. One way to understand this behavior is to view TMLE-PR through the same decomposition underlying A-TMLE. When $W$ is highly predictive of $S$, the efficient influence curve for estimating the bias parameter $\Psi^\#$ depends on inverse weighting of the subgroup membership mechanism $\Pi$ in a nonparametric model. A-TMLE explicitly regularizes this component, which can lead to improved finite-sample efficiency relative to a standard TMLE.

\section{Case study: LEADER trial subgroup analysis}\label{sec:LEADER}

To demonstrate the practical utility of our approach, we applied A-TMLE to estimate subgroup-specific treatment effects using data from the LEADER cardiovascular outcome trial \citep{marso2013design,marso_liraglutide_2016}. LEADER was a large, multinational, randomized, double-blind, placebo-controlled trial designed to evaluate the cardiovascular safety of liraglutide in patients with type 2 diabetes mellitus (T2DM) at high risk for cardiovascular events. The trial population consisted of T2DM patients 50 years of age or older with at least one pre-existing cardiovascular condition, with an HbA1c level of 7.0\% or greater, and who were either (1) treatment naive, or (2) treated with oral anti-hyperglycemic agents and/or insulin. The trial included $n=9,340$ patients, $4,669$ in the liraglutide treatment group and $4,672$ in the placebo control group. Pre-specified subgroups consisted of persons self-identifying as American Indian or Alaska Native ($n=11$), Asian ($n=936$), Black or African American ($n=777$), Native Hawaiian or other Pacific Islander ($n=8$), Other ($n=370$), or White ($n=7238$). For the purposes of this subgroup analysis, we focused on effects in the Asian and Black subgroups, since the majority of trial participants are White, and the remaining subgroups have very small sample sizes. Mean exposure time in the trial was 3.5 years. 

Our goal was to estimate subgroup-specific risk differences in MACE between patients receiving the active intervention and those in the control arm across racial subgroups. Some subgroups are relatively underrepresented in the trial, leading to high variance and reduced statistical power for conventional subgroup-specific estimators. We apply A-TMLE to borrow strength from the remainder of the trial population while correcting for subgroup-specific differences, effectively treating observations outside each subgroup as auxiliary data.

We estimated the marginal average treatment effect comparing treatment arm to control arm on MACE events at 1, 1.5, and 2 years. We adjusted for the following baseline covariates: age, sex, diabetes duration, previous/current use of insulin, antihypertensive therapy use, eGFR, BMI, HbA1c, LDL, HDL and evidence of prior cardiovascular event. Details on the configurations of the A-TMLE and AIPW estimator are available in Appendix \ref{app:implementation_details}. The results are presented in Figure \ref{fig:LEADER_results}. Across the first 730 days of follow-up, the subgroup analyses generally suggest that liraglutide is associated with a lower event risk than placebo in both Asian and Black subgroups, but the strength and stability of that signal differ by estimator. In the Black or African American subgroup, the A-TMLE estimates are remarkably stable, at about a 2-percentage-point reduction at 365, 540, and 730 days, with 95\% confidence intervals excluding zero at all three time points. The AIPW estimates are similar in the direction of the effect but more variable, ranging from roughly a 1.8- to 3.6-percentage-point reduction, with statistical uncertainty large enough that only the 540-day estimate excludes zero. In the Asian subgroup, the contrast between estimators is larger. The A-TMLE again suggests a modest but consistent protective effect of about 1.5 to 1.6 percentage points at each time horizon, whereas the AIPW estimates are close to zero and quite imprecise, with much larger confidence intervals covering the null.

\begin{figure}[!ht]
\centering
\includegraphics[width=0.8\linewidth]{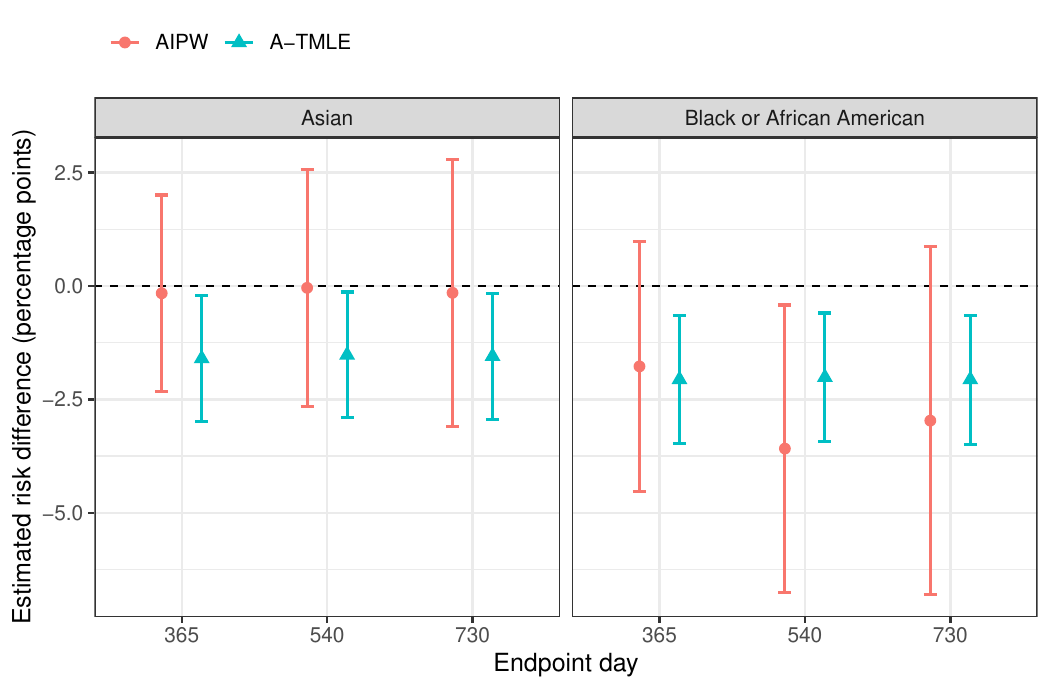}
\caption{Estimated risk differences comparing liraglutide versus placebo for the primary endpoint in the LEADER trial within the Asian and Black subgroups, evaluated at 365, 540, and 730 days of follow-up. Each point denotes the subgroup-specific estimated risk difference, expressed in percentage points, and vertical bars denote 95\% confidence intervals. Results are shown side by side for the AIPW estimator and the A-TMLE estimator. Negative values indicate a lower estimated event risk under liraglutide than under placebo.}
\label{fig:LEADER_results}
\end{figure}

\section{Discussion}\label{sec:discussion}

We proposed two estimators for subgroup-specific treatment effect estimation that borrow strength from observations outside the subgroup of interest while preserving the target estimand defined over the subgroup covariate distribution, $P_{W\mid S=1}$. By viewing the remainder of the RCT population as an auxiliary source, we obtain an estimator that remains grounded in the internal validity of the trial while potentially achieving substantial efficiency gains relative to subgroup-only analyses. Both the simulation study and the LEADER case study suggest that this approach can improve precision, and our LEADER analysis provides subgroup-specific estimates in racial groups for which the original trial was not powered to draw definitive conclusions.

A key feature of A-TMLE is the decomposition $\Psi(P)=\tilde{\Psi}(P)-\Psi^\#(P)$, which makes explicit that the benefit of borrowing information depends on the size and complexity of the subgroup-specific bias term. This is both a strength and a limitation. When the bias function is small or well approximated by a relatively low-complexity working model, the projected bias correction can be estimated stably and the resulting A-TMLE can be much more efficient than conventional subgroup-specific estimators. However, when the subgroup-by-treatment interaction is highly complex, projection bias may be non-negligible, and the efficiency gains from borrowing strength may be reduced \citep{vanderlaan2024atmle}. Thus, the practical value of the method depends not only on subgroup sample size, but also on how well the bias structure can be represented by the selected working model.

A related limitation concerns the current working-model selection strategy for the bias component. In the RCT setting considered here, the bias estimand is driven by treatment-dependent subgroup differences, that is, by subgroup-by-treatment interaction. To avoid selecting a degenerate working model with zero estimated bias simply because the treatment variable is omitted, we currently restrict model selection to a class of working models that always include the treatment variable $A$. While operationally convenient, this requirement is somewhat unnatural, as the object driving $\Psi^\#$ is not the main effect of $A$ itself, but rather the treatment-dependent component of $\tau_S(W,A)$. A more principled approach would be to develop a targeted working-model selection strategy that more directly prioritizes subgroup-by-treatment interaction structure. We view such targeted selection as an important direction and note that related efforts are already underway.

There are several additional limitations. First, for simplicity we assumed no loss to follow-up and reduced the time-to-event outcome to a fixed time horizon risk. Extensions to right-censoring and other trial complications remain important open problems. Note that for right-censored failure-time data, a simple (non-efficient) approach is to use inverse probability of censoring weighting with A-TMLE. For A-TMLE, unlike the treatment mechanism $\bar g_0$, which is known by design in the RCT, the subgroup-membership mechanism still must be estimated from the data. In our setting this is the quantity $\bar{\Pi}_0(1\mid W)=P_0(S=1\mid W)$. For the asymptotic results to kick in, this mechanism needs to be estimated sufficiently well. This may be challenging because subgroup membership can itself be a complicated function of baseline covariates, especially when the subgroup is characterized by a combination of demographic, clinical, and laboratory features. Thus, even though the subgroup setting is more favorable than general RCT-RWD integration in several respects, this particular nuisance component remains nontrivial. Relatedly, if $W$ is highly predictive of $S$, then $\bar{\Pi}_0(1\mid W)$ may be close to 0 or 1 over substantial parts of the covariate space. In that regime, the coefficients indexing the working model for $\tau_S(W,A)$ may become only weakly identified. In such settings, A-TMLE effectively places less emphasis on learning the working models in regions with limited overlap, focusing instead on areas of the covariate space where overlap is adequate. This behavior can be either advantageous or problematic. On one hand, consider a baseline covariate that strongly predicts subgroup membership, such as BMI. For strata with high BMI, if subgroup membership has little impact on the mean outcome, then limited attention to these regions is beneficial: there is little bias to correct, and reduced modeling effort can lead to lower variance. On the other hand, if individuals with high BMI exhibit substantially different mean outcomes depending on subgroup membership, then this becomes problematic. In such cases, A-TMLE may under-emphasize these regions and fail to adequately correct for bias, effectively leaving bias unaddressed.

Despite these limitations, we believe this work provides a useful and principled framework for improving subgroup analyses using only data already collected within the trial itself. More broadly, it suggests that data integration ideas can be fruitfully repurposed within a single RCT by treating non-subgroup participants as an auxiliary source and explicitly modeling the bias induced by borrowing information across populations. Future work will focus on more targeted working-model selection strategies, extensions to settings with right-censoring and more complicated data structures, and broader empirical evaluation in additional trial applications.

\newpage
\bibliography{references}

\newpage

\appendix

\section{Notation}\label{app:notation}

\begingroup
\small
\setstretch{1.05}
\renewcommand{\arraystretch}{1.15}
\setlength{\LTcapwidth}{\textwidth}

\begin{longtable}{
  @{}
  >{\raggedright\arraybackslash}p{0.30\textwidth}
  >{\raggedright\arraybackslash}
    p{\dimexpr0.70\textwidth-2\tabcolsep\relax}
  @{}
}
\caption{Notation used throughout the paper.}
\label{tab:notation}\\

\toprule
\textbf{Notation} & \textbf{Definition or interpretation} \\
\midrule
\endfirsthead

\multicolumn{2}{@{}l@{}}{
  \tablename~\thetable\ (continued)
}\\
\toprule
\textbf{Notation} & \textbf{Definition or interpretation} \\
\midrule
\endhead

\midrule
\multicolumn{2}{r@{}}{\textit{Continued on the next page}}\\
\endfoot

\bottomrule
\endlastfoot

\multicolumn{2}{@{}l@{}}{
  \textbf{Data and probability notation}
}\\*

$O=(S,W,A,Y)$
&
Observed data for one trial participant under the simplified fixed-horizon outcome formulation.
\\

$S$
&
Subgroup indicator: $S=1$ denotes the subgroup of interest and $S=0$ denotes the remainder of the trial population.
\\

$W$, $w$
&
Vector of baseline covariates and a possible realization, respectively.
\\

$A$, $a$
&
Binary treatment assignment and a possible treatment level: $1$ for treatment and $0$ for control.
\\

$T$, $C$, $\widetilde T$, $\Delta$
&
Event time, censoring time, observed follow-up time $\widetilde T=\min(T,C)$, and indicator that the event is observed, respectively.
\\

$\tau$, $Y$, $Y_a$
&
Evaluation time, binary outcome $Y=I(T\leq\tau)$, and potential
outcome under assignment to treatment level $a$, respectively.
\\

$I(\cdot)$
&
Indicator function, equal to one when its argument is true and zero otherwise.
\\

$P^F$, $P_0^F$, $\mathcal M^F$
&
A generic joint distribution of $(S,W,A,Y_0,Y_1)$, the true full-data distribution, and the full-data model, respectively. Under consistency, $P^F$ induces $P$ as the law of $(S,W,A,Y_A)$.
\\

$P$, $P_0$, $\mathcal M$
&
A generic observed-data distribution, the true observed-data distribution, and the statistical model containing these distributions, respectively.
\\

$E_P$, $E_0$, $Pf$
&
Expectation under $P$, expectation under $P_0$, and the shorthand $Pf=\int f(o)\,dP(o)=E_P\{f(O)\}$, respectively.
\\

$P_n$, $E_n$
&
Empirical distribution and empirical expectation: $P_nf=E_nf=n^{-1}\sum_{i=1}^n f(O_i)$.
\\

$P_W$, $P_{W\mid S=1}$
&
Marginal covariate distribution in the full trial population and conditional covariate distribution within the subgroup under $P$, respectively. The corresponding distributions under $P^F$ are denoted by $P_W^F$ and $P_{W\mid S=1}^F$; these agree with their observed-data counterparts when $P$ is induced by $P^F$.
\\

$p_P$, $p_0$
&
Subgroup prevalence under $P$ and $P_0$: $p_P=P(S=1)$ and $p_0=P_0(S=1)$.
\\

Subscripts $P$, $0$, and $n$
&
Indicate evaluation under a generic distribution $P$, evaluation under the true distribution $P_0$, and a sample-based estimate, respectively; for example, $Q_P$, $Q_0$, and $Q_n$.
\\

Superscripts $\star$ and $(k)$
&
Indicate a targeted update and the $k$th iteration of an estimation algorithm, respectively.
\\

Working-model subscript $w$
&
Indicates association with a working model, as in $\mathcal M_w$ and $Q_{w,P}$; it is distinct from the covariate value $w$ appearing as a function argument.
\\

\midrule
\multicolumn{2}{@{}l@{}}{
  \textbf{Outcome regressions and assignment mechanisms}
}\\*

$Q_P(s,w,a)$ & $E_P(Y\mid S=s,W=w,A=a)$: outcome regression conditional on subgroup membership, baseline covariates, and treatment.
\\

$\bar Q_P(w,a)$
&
$E_P(Y\mid W=w,A=a)$: pooled outcome regression, without conditioning on subgroup membership.
\\

$\theta_P(w)$
&
$E_P(Y\mid W=w)$: outcome regression conditional only on baseline covariates, used to residualize the outcome when learning the pooled treatment-effect function.
\\

$g_P(a\mid s,w)$
&
$P(A=a\mid S=s,W=w)$: treatment-assignment probability conditional on subgroup membership and baseline covariates.
\\

$\bar g_P(a\mid w)$
&
$P(A=a\mid W=w)$: treatment-assignment probability conditional only on baseline covariates.
\\

$\Pi_P(s\mid w,a)$
&
$P(S=s\mid W=w,A=a)$: subgroup-membership probability conditional on baseline covariates and treatment.
\\

$\bar\Pi_P(s\mid w)$
&
$P(S=s\mid W=w)$: subgroup-membership probability conditional only on baseline covariates. Under the trial randomization, $\Pi_0(s\mid w,a)=\bar\Pi_0(s\mid w)$.
\\

$\tau_{A,P}(w)$
&
$\bar Q_P(w,1)-\bar Q_P(w,0)$: pooled conditional treatment-effect contrast.
\\

$\tau_{S,P}(w,a)$
&
$Q_P(1,w,a)-Q_P(0,w,a)$: conditional outcome contrast between subgroup and non-subgroup participants, also called the bias function.
\\

\midrule
\multicolumn{2}{@{}l@{}}{
  \textbf{Target parameter and its decomposition}
}\\*

$\Psi^F(P^F)$
&
Full-data subgroup causal target: $E_{P^F}(Y_1-Y_0\mid S=1)$. Equivalently, the conditional treatment effect is averaged over $P_{W\mid S=1}^F$ by the law of iterated expectations.
\\

$\Psi(P)$
&
Observed-data subgroup statistical target: $E_P\left[Q_P(1,W,1)-Q_P(1,W,0)\mid S=1\right]$. Under consistency and assumptions A1--A2, $\Psi^F(P_0^F)=\Psi(P_0)$.
\\

$\widetilde\Psi(P)$
&
Pooled-regression parameter: $E_P\left[\tau_{A,P}(W)\mid S=1\right]$. The regression pools over subgroup membership, but the outer expectation remains over the subgroup covariate distribution.
\\

$\Psi^\#(P)$
&
Bias parameter, defined by $\Psi^\#(P)=\widetilde\Psi(P)-\Psi(P)$, equivalently, $E_P\bigl[\Pi_P(0\mid W,0)\tau_{S,P}(W,0)-\Pi_P(0\mid W,1)\tau_{S,P}(W,1)\mid S=1\bigr]$.
\\

$\Psi_n^{\mathrm{TMLE\mbox{-}PR}}$
&
TMLE with pooled initial outcome regression and subgroup-specific targeting:
$$
\frac{P_n\left[S\{Q_n^\star(1,W,1)-Q_n^\star(1,W,0)\}\right]}{P_n(S=1)}.
$$
\\

\midrule
\multicolumn{2}{@{}l@{}}{
  \textbf{Working models and projection parameters}
}\\*

$\phi_A(W)$, $\phi_S(W,A)$
&
Vectors of basis functions for the pooled treatment-effect function and the conditional subgroup contrast, respectively. Their components are denoted by $\phi_{A,j}$ and $\phi_{S,j}$.
\\

$\tau_{A,\beta_A}(W)$,
$\tau_{S,\beta_S}(W,A)$
&
Working-model functions indexed by coefficient vectors
$\beta_A$ and $\beta_S$. For the linear basis representations,
$$
\begin{aligned}
\tau_{A,\beta_A}(W)
  &=\sum_j\beta_A(j)\phi_{A,j}(W),\\
\tau_{S,\beta_S}(W,A)
  &=\sum_j\beta_S(j)\phi_{S,j}(W,A).
\end{aligned}
$$
\\

$\mathcal T_{A,w}$,
$\mathcal T_{S,w}$
&
Working-model classes for $\tau_{A,P}$ and $\tau_{S,P}$, respectively, indexed by coefficients in $\mathcal B_A$ and $\mathcal B_S$.
\\

$\bar{\mathcal Q}_w$,
$\mathcal Q_w$
&
Semiparametric outcome-regression working models induced by restricting the treatment-effect function to $\mathcal T_{A,w}$ or the subgroup-contrast function to $\mathcal T_{S,w}$, respectively.
\\

$\beta_A(P)$, $\beta_{A,P}$
&
Equivalent notation for the coefficients of the weighted
squared-error projection of $\tau_{A,P}$ onto
$\mathcal T_{A,w}$, using weight
$\bar g_P(1\mid W)\bar g_P(0\mid W)$.
\\

$\beta_S(P)$, $\beta_{S,P}$
&
Equivalent notation for the coefficients of the weighted
squared-error projection of $\tau_{S,P}$ onto
$\mathcal T_{S,w}$, using weight
$\Pi_P(1\mid W,A)\Pi_P(0\mid W,A)$.
\\

$\bar Q_{w,P}(W,A)$
&
Projected pooled outcome regression: $\theta_P(W)+\{A-\bar g_P(1\mid W)\}\tau_{A,\beta_A(P)}(W)$.
\\

$Q_{w,P}(S,W,A)$
&
Projected subgroup-specific outcome regression: $\bar Q_P(W,A)+\{S-\Pi_P(1\mid W,A)\}\tau_{S,\beta_S(P)}(W,A)$.
\\

$\mathcal M_{A,w}$,
$\mathcal M_{S,w}$,
$\mathcal M_w$
&
Restricted statistical models:
$$
\begin{aligned}
\mathcal M_{A,w}
&=\{P\in\mathcal M:\tau_{A,P}\in\mathcal T_{A,w}\},\\
\mathcal M_{S,w}
&=\{P\in\mathcal M:\tau_{S,P}\in\mathcal T_{S,w}\},\\
\mathcal M_w
&=\mathcal M_{A,w}\cap\mathcal M_{S,w}.
\end{aligned}
$$
\\

$\widetilde\Psi_{\mathcal M_{A,w}}(P)$
&
Pooled projection parameter: $E_P\{\tau_{A,\beta_A(P)}(W)\mid S=1\}$.
\\

$\Psi^\#_{\mathcal M_{S,w}}(P)$
&
Bias projection parameter obtained by replacing $\tau_{S,P}$ in $\Psi^\#(P)$ with $\tau_{S,\beta_S(P)}$, while retaining $\Pi_P$ and the subgroup covariate distribution.
\\

$\Psi_{\mathcal M_w}(P)$
&
Projected subgroup target: $\widetilde\Psi_{\mathcal M_{A,w}}(P)-\Psi^\#_{\mathcal M_{S,w}}(P)$. This need not equal $\Psi(P)$ for an arbitrary $P$ outside the working model.
\\

$m_{A,P}(W)$, $b_P(W)$
&
Conditional integrands of the pooled and bias projection parameters:
$$
\begin{aligned}
m_{A,P}(W)
&=\tau_{A,\beta_A(P)}(W),\\
b_P(W)
&=\Pi_P(0\mid W,0)\tau_{S,\beta_S(P)}(W,0)-\Pi_P(0\mid W,1)\tau_{S,\beta_S(P)}(W,1).
\end{aligned}
$$
\\

$m_P(W)$
&
In Section \ref{sec:tmle_pr}, the unprojected subgroup contrast $Q_P(1,W,1)-Q_P(1,W,0)$. In Section \ref{sec:atmle}, the projected integrand $m_{A,P}(W)-b_P(W)$.
\\

\midrule
\multicolumn{2}{@{}l@{}}{
  \textbf{Influence functions and exact remainders}
}\\*

$D_{\Psi,P}$,
$D_{\mathcal M_w,P}$
&
Canonical gradients, also called efficient influence functions, of $\Psi$ and $\Psi_{\mathcal M_w}$, respectively, in the models specified in the corresponding results.
\\

$D_{\widetilde\Psi,\mathcal M_{A,w},P}$,
$D_{\Psi^\#,\mathcal M_{S,w},P}$
&
Canonical gradients of the pooled and bias projection parameters, respectively. Their difference is $D_{\mathcal M_w,P}$.
\\

$D_{A,\beta_A,P}$,
$D_{S,\beta_S,P}$
&
Vector-valued influence functions for the projection
coefficients $\beta_A(P)$ and $\beta_S(P)$, respectively.
\\

$I_{A,P}$, $I_{S,P}$
&
Weighted basis Gram matrices:
\[
\begin{aligned}
I_{A,P}&=E_P\bigl[\bar g_P(1\mid W)\bar g_P(0\mid W)\phi_A(W)\phi_A(W)^\top\bigr],\\
I_{S,P}&=E_P\bigl[\Pi_P(1\mid W,A)\Pi_P(0\mid W,A)\phi_S(W,A)\phi_S(W,A)^\top\bigr].
\end{aligned}
\]
\\

$R_\Psi(P,P_0)$,
$R_{\mathcal M_w}(P,P_0)$
&
Exact remainders for the subgroup target and its projection. For a parameter $\Phi$ with canonical gradient $D_{\Phi,P}$, the convention is
\[
R_\Phi(P,P_0)
=\Phi(P)-\Phi(P_0)+P_0D_{\Phi,P}.
\]
\\

$\Phi_h(P)$, $N_h(P)$
&
Generic conditional-mean parameter and its numerator:
\[
\begin{aligned}
\Phi_h(P)&=E_P\{h_P(W)\mid S=1\}=\frac{N_h(P)}{p_P},\\
N_h(P)&=E_P\{S h_P(W)\}.
\end{aligned}
\]
\\

$N_\Psi(P)$, $N_A(P)$,
$N_S(P)$, $N_{\mathcal M_w}(P)$
&
Numerators of $\Psi(P)$,
$\widetilde\Psi_{\mathcal M_{A,w}}(P)$,
$\Psi^\#_{\mathcal M_{S,w}}(P)$, and
$\Psi_{\mathcal M_w}(P)$, respectively.
Each equals the corresponding parameter multiplied by $p_P$;
$N_{\mathcal M_w}=N_A-N_S$.
\\

$\mathcal M_0$,
$\Psi_{\mathcal M_0}$,
$D_{\mathcal M_0,P}$
&
Oracle model containing $P_0$, its projection parameter, and
the corresponding canonical gradient.
\\

$\Pi_{\mathcal M_w}$,
$\Pi_{\mathcal M_0}$,
$P_{0,n}$, $\widetilde P_{0,n}$
&
Model-projection operators and projected distributions:
\[
\begin{aligned}
P_{0,n}&=\Pi_{\mathcal M_w}(P_0),\\
\widetilde P_{0,n}&=\Pi_{\mathcal M_0}(P_{0,n}).
\end{aligned}
\]
These projection operators are distinct from the
subgroup-membership mechanism $\Pi_P$.
\\

$T_{\mathcal M_w}(\widetilde P_{0,n})$,
$\widetilde D_{\mathcal M_w,\widetilde P_{0,n}}$
&
Tangent space of the working model at $\widetilde P_{0,n}$ and projection of $D_{\mathcal M_0,\widetilde P_{0,n}}$ onto that tangent space, respectively.
\\

$\sigma_{w,n}^2$
&
Variance of the working-model influence function at the truth: $\sigma_{w,n}^2=P_0D_{\mathcal M_w,P_0}^2$. The index $n$ allows for a data-adaptively selected working model.
\\

\midrule
\multicolumn{2}{@{}l@{}}{
  \textbf{Estimation, targeting, and working-model selection}
}\\*

$P_n^\star$, $P_n^{(k)}$
&
Targeted estimate of the distribution or its relevant nuisance components, and the corresponding estimate at iteration $k$. These are distinct from the empirical distribution $P_n$.
\\

$H_n(W,A)$, $\epsilon_n$
&
Clever covariate and fitted fluctuation coefficient for the TMLE-PR logistic update:
\[
\operatorname{logit}Q_n^\star(1,W,A)=\operatorname{logit}Q_n(1,W,A)+\epsilon_n H_n(W,A).
\]
\\

$\widetilde L_{\theta,\bar g}$,
$L^\#_{\bar Q,\Pi}$
&
Residualized squared-error losses used to estimate $\tau_A$ and $\tau_S$, respectively:
\[
\begin{aligned}
\widetilde L_{\theta,\bar g}(O)(\tau)
 &=\{Y-\theta(W)-(A-\bar g(1\mid W))\tau(W)\}^2,\\
L^\#_{\bar Q,\Pi}(O)(\tau)
 &=\{Y-\bar Q(W,A)-(S-\Pi(1\mid W,A))\tau(W,A)\}^2.
\end{aligned}
\]
\\

$C_{A,n}$, $C_{S,n}$
&
Coefficient-space targeting directions for the pooled and bias components. They are constructed from the inverse weighted Gram matrices and the corresponding parameter derivatives.
\\

$H_{A,n}$, $H_{S,n}$
&
Scalar covariates for the one-dimensional coefficient fluctuations, formed from the residualized basis vectors and the directions $C_{A,n}$ and $C_{S,n}$, respectively.
\\

$R_{A,n}$, $R_{S,n}$
&
Outcome residuals from the current projected pooled and subgroup outcome regressions, respectively. These are regression residuals, distinct from the exact remainders $R_\Psi$ and $R_{\mathcal M_w}$.
\\

$\epsilon^\star$, $\gamma^\star$
&
Fitted scalar fluctuation coefficients for updating $\beta_A$ and $\beta_S$, respectively; their empirical forms are ratios of residual--covariate products to squared fluctuation covariates.
\\

$\varepsilon_n^{(k)}$
&
Fitted logistic fluctuation coefficient used to update the subgroup-membership mechanism $\Pi_n^{(k)}$.
\\

$s_n$, $\sigma_n^{(k)}$, $K$
&
Stopping threshold, estimated standard deviation of the influence function at iteration $k$, and final iteration index. The algorithm uses $s_n=\sigma_n^{(k)}/(\sqrt n\log n)$.
\\

$\phi_j$, $J$, $f_\beta$
&
Generic HAL basis functions, number of basis functions, and basis expansion $f_\beta=\sum_{j=1}^J\beta_j\phi_j$.
\\

$C_n$, $M_{A,n}$, $M_{S,n}$
&
Bounds on the coefficient $\ell_1$ norm controlling the complexity of a generic HAL fit, the pooled-effect HAL fit, and the subgroup-contrast HAL fit, respectively.
\\

$\mathcal R_n$, $\mathcal R_0$,
$\mathcal D(\mathcal R)$
&
Data-adaptively selected basis index set, fixed oracle index set, and working model generated by a specified index set $\mathcal R$, respectively.
\\

$\mathcal R_{A,n}$,
$\mathcal R_{S,n}$
&
Selected basis index sets for the pooled treatment-effect and subgroup-contrast working models, respectively.
\\

$\mathcal S_m$, $\beta_s$
&
Collection of nonempty coordinate subsets of size at most $m$ used in HAL subspace screening, and the coefficient block for basis functions using coordinate subset $s$, respectively. Here, $s$ indexes coordinates rather than subgroup membership.
\\

$\lambda$, $w_s$, $p_s$
&
Group-lasso penalty parameter, weight for coordinate group $s$, and number of basis functions in that group, respectively; one choice is $w_s=\sqrt{p_s}$.
\\

\midrule
\multicolumn{2}{@{}l@{}}{
  \textbf{Simulation notation}
}\\*

$\lambda_j(t\mid S,A,W)$,
$\Lambda_j(t\mid S,A,W)$
&
Conditional hazard and cumulative hazard under simulation scenario $j$, respectively.
\\

$F_j(t\mid S,A,W)$,
$\psi_j(t)$
&
Conditional event risk by time $t$ under scenario $j$ and the corresponding true subgroup risk difference:
\[
\psi_j(t)
=
E\left[
F_j(t\mid 1,1,W)-F_j(t\mid 1,0,W)
\mid S=1
\right].
\]
\\

\end{longtable}
\endgroup

\section{Proof of lemmas}\label{app:can_grads}

\subsection{Proof of Lemma \ref{lem:can_grad_psi_tilde}}

\begin{proof}
First, note that by the law of total expectation, we can express our target parameter $\tilde{\Psi}_{\mathcal{M}_{A,w}}$ as
$$
\tilde{\Psi}_{\mathcal{M}_{A,w}}(P)=\frac{1}{P(S=1)}E[I(S=1)\cdot\tau_{A,\beta_P}(W)].
$$
We take the pathwise derivative of $\tilde{\Psi}_{\mathcal{M}_{A,w}}$ along a path $P_\epsilon$ through $P$ at $\epsilon=0$ with a corresponding score
$$
s(O)=\frac{d}{d\epsilon}\log dP_{\epsilon}(O)\eval_{\epsilon=0}.
$$
We have that
$$
\frac{d}{d\epsilon}\tilde{\Psi}_{\mathcal{M}_{A,w}}(P_\epsilon)\eval_{\epsilon=0}=\frac{d}{d\epsilon}\left[\frac{1}{P_{\epsilon}(S=1)}E_{P_\epsilon}\left[S\cdot\tau_{A,\beta_{P_\epsilon}}(W)\right]\right]_{\epsilon=0}.
$$
Define
$$
\mathbf{A}_\epsilon=P_\epsilon(S=1)\quad\text{and}\quad \mathbf{B}_\epsilon=\int S\cdot\tau_{A,\beta_{P_\epsilon}}(W)dP_{\epsilon}(O).
$$
Then, we have
\begin{align*}
\frac{d}{d\epsilon}\tilde{\Psi}_{\mathcal{M}_{A,w}}(P_\epsilon)\eval_{\epsilon=0}&=-\frac{1}{\mathbf{A}^2}\cdot\frac{d}{d\epsilon}\mathbf{A}_\epsilon\eval_{\epsilon=0}\mathbf{B}+\frac{1}{\mathbf{A}}\cdot\frac{d}{d\epsilon}\mathbf{B}_\epsilon\eval_{\epsilon=0}\\
&=\frac{1}{\mathbf{A}}\left[\frac{d}{d\epsilon}\mathbf{B}_\epsilon-\tilde{\Psi}_{\mathcal{M}_{A,w}}(P)\cdot\frac{d}{d\epsilon}\mathbf{A}_\epsilon\right]_{\epsilon=0}.
\end{align*}
Recall the influence curve of $P(S=1)$ is $I(S=1)-P(S=1)$. Therefore, we have
$$
\frac{d}{d\epsilon}\mathbf{A}_{\epsilon}\eval_{\epsilon=0}=E_P[I(S=1)s(O)].
$$
Note also that
$$
\frac{d}{d\epsilon}\tau_{A,\beta_{P_\epsilon}}(W)\eval_{\epsilon=0}=\frac{d}{d\epsilon}\left[\beta_{P_\epsilon}^\top\phi_A(W)\right]_{\epsilon=0}=\phi_A^\top(W)\frac{d}{d\epsilon}\beta_{P_{\epsilon}}\eval_{\epsilon=0}=\phi^\top_A(W)E_P[D_{A,\beta,P}(O)s(O)].
$$
Therefore,
\begin{align*}
\frac{d}{d\epsilon}\mathbf{B}_{\epsilon}\eval_{\epsilon=0}&=E_P\left\{S\cdot\tau_{A,\beta_P}(W)s(O)+S\cdot\phi^\top_A(W)E_P\left[D_{A,\beta,P}(O)s(O)\right]\right\}\\
&=E_P\left\{\left[S\cdot \tau_{A,\beta_P}(W)+E_P[S\phi_A(W)]^\top D_{A,\beta,P}(O)\right]s(O)\right\}.
\end{align*}
Finally, we have that
$$
\frac{d}{d\epsilon}\tilde{\Psi}_{\mathcal{M}_{A,w}}(P_\epsilon)\eval_{\epsilon=0}=E_P\left\{\left[\frac{S}{P(S=1)}\left[\tau_{A,\beta_P}(W)-\tilde{\Psi}_{\mathcal{M}_{A,w}}(P)\right]+\frac{1}{P(S=1)}E_P[S\phi_A(W)]^\top D_{A,\beta,P}(O)\right]\cdot s(O)\right\}.
$$
Hence, the canonical gradient of $\tilde{\Psi}_{\mathcal{M}_{A,w}}$ at $P$ is given by
$$
D_{\tilde{\Psi},\mathcal{M}_{A,w},P}(O)=\frac{S}{P(S=1)}\left[\tau_{A,\beta_P}(W)-\tilde{\Psi}_{\mathcal{M}_{A,w}}(P)\right]+\frac{1}{P(S=1)}E_P[S\phi_A(W)]^\top D_{A,\beta,P}(O).
$$
\end{proof}

\subsection{Proof of Lemma \ref{lem:can_grad_psi_pound}}

\begin{proof}
First, we again express our target parameter $\Psi_{\mathcal{M}_{S,w}}^\#$ as
$$
\Psi_{\mathcal{M}_{S,w}}^\#(P)=\frac{1}{P(S=1)}E_P\bigg\{S\left[\Pi_P(0\mid W,0)\tau_{S,\beta_P}(W,0)-\Pi_P(0\mid W,1)\tau_{S,\beta_P}(W,1)\right]\bigg\}.
$$
We define the control-arm portion and treatment-arm portion of the target parameter as
\begin{align*}
\Psi^\#_{\mathcal{M}_{S,w},A_0}(P)&=\frac{1}{P(S=1)}E_P\bigg\{S[\Pi_P(0\mid W,0)\tau_{S,\beta_P}(W,0)]\bigg\},\\
\Psi^\#_{\mathcal{M}_{S,w},A_1}(P)&=\frac{1}{P(S=1)}E_P\bigg\{S[\Pi_P(0\mid W,1)\tau_{S,\beta_P}(W,1)]\bigg\},
\end{align*}
respectively, so that $\Psi^\#_{\mathcal{M}_{S,w}}(P)=\Psi^\#_{\mathcal{M}_{S,w},A_0}(P)-\Psi^\#_{\mathcal{M}_{S,w},A_1}(P)$. We focus on the control-arm portion as the treatment-arm portion is an exact mirroring of the same derivation. Similar to the previous proof, let
$$
\mathbf{A}_\epsilon=P_\epsilon(S=1)\quad\text{and}\quad\mathbf{C}_\epsilon=\int S\cdot [\Pi_{P_\epsilon}(0\mid W,0)\tau_{S,\beta_{P_\epsilon}}(W,0)]dP_{\epsilon}(O).
$$
Note that the influence curve of $\Pi_P(0\mid w,a)$ is given by
$$
\frac{I(W=w,A=a)}{P(w,a)}(I(S=0)-\Pi_P(0\mid w,a)).
$$
So, we have
\begin{align*}
\frac{d}{d\epsilon}\Pi_{P_\epsilon}(0\mid W,0)\eval_{\epsilon=0}&=E_P\left\{\frac{1-A}{g_P(0\mid W)}(I(S=0)-\Pi_P(0\mid W,0))s(O)\bigg\rvert W\right\}\\
&=-E_P\left\{\frac{1-A}{g_P(0\mid W)}(S-\Pi_P(1\mid W,0))s(O)\bigg\rvert W\right\}
\end{align*}
As derived in Lemma 7 of our previous A-TMLE data fusion paper, the canonical gradient of $\beta$ at $P$ is given by
$$
D_{S,\beta,P}(O)=I_{S,P}^{-1}(S-\Pi_P(1\mid W,A))\phi_{S}(W,A)(Y-Q_{w,P}(S,W,A)).
$$
So, we have
$$
\frac{d}{d\epsilon}\tau_{S,\beta_{P_\epsilon}}(W,0)\eval_{\epsilon=0}=\phi_S^\top(W,0)E_P[D_{S,\beta,P}(O)s(O)].
$$
Therefore, we have
\begin{align*}
\frac{d}{d\epsilon}\mathbf{C}_\epsilon\eval_{\epsilon=0}&=E_P\bigg\{\bigg(S\cdot\Pi_P(0\mid W,0)\tau_{S,\beta_P}(W,0)-E_P[S\cdot\Pi_P(0\mid W,0)\tau_{S,\beta_P}(W,0)]\bigg)s(O)\bigg\}\\
&-E_P\bigg\{\left(\bar{\Pi}_P(1\mid W)\tau_{S,\beta_P}(W,0)\frac{1-A}{\bar{g}_P(0\mid W)}(S-\Pi_P(1\mid W,0))\right)s(O)\bigg\}\\
&+E_P\bigg\{\bigg(E_P[S\cdot\Pi_P(0\mid W,0)\phi_S(W,0)]^\top D_{S,\beta,P}(O)\bigg)s(O)\bigg\}.
\end{align*}
Recall from the previous proof that we have
$$
\frac{d}{d\epsilon}\Psi_{\mathcal{M}_{S,w},A_0}^\#(P_\epsilon)\eval_{\epsilon=0}=\frac{1}{\mathbf{A}}\left[\frac{d}{d\epsilon}\mathbf{C}_\epsilon-\Psi_{\mathcal{M}_{S,w},A_0}^\#(P)\cdot\frac{d}{d\epsilon}\mathbf{A}_\epsilon\right]_{\epsilon=0}.
$$
Hence, the canonical gradient of $\Psi_{\mathcal{M}_{S,w},A_0}^\#$ at $P$ is
\begin{align*}
D_{\Psi^\#,\mathcal{M}_{S,w},A_0,P}(O)&=\frac{S}{P(S=1)}\left[\Pi_P(0\mid W,0)\tau_{S,\beta_P}(W,0)-\Psi^\#_{\mathcal{M}_{S,w},A_0}(P)\right]\\
&-\frac{1}{P(S=1)}\bar{\Pi}_P(W)\tau_{S,\beta_P}(W,0)\frac{1-A}{g_P(0\mid W)}(S-\Pi_P(1\mid W,0))\\
&+\frac{1}{P(S=1)}E_P\left[S\cdot\Pi_P(0\mid W,0)\phi_S(W,0)\right]^\top D_{S,\beta,P}(O).
\end{align*}
The derivation for $D_{\Psi^\#,\mathcal{M}_{S,w},A_1,P}(O)$ mirrors this. We therefore arrive at the desired canonical gradient presented in the lemma.
\end{proof}

\section{Exact remainder derivations}\label{app:exact_remainder}

\subsection{Proof of Lemma \ref{lem:ratio_remainder_conditional_mean}}

\begin{proof}
Since $\Phi_h(P)=N_h(P)/p_P$, by the quotient-rule,
$$
D_{\Phi_h,P}(O)=\frac{D_{N_h,P}(O)-\Phi_h(P)\{S-p_P\}}{p_P}.
$$
Therefore,
\begin{align*}
R_{\Phi_h}(P,P_0)&=\frac{N_h(P)}{p_P}-\frac{N_h(P_0)}{p_0}+\frac{1}{p_P}P_0D_{N_h,P}-\frac{\Phi_h(P)}{p_P}P_0\{S-p_P\}\\
&=\frac{1}{p_P}\Bigl\{N_h(P)-N_h(P_0)+P_0D_{N_h,P}\Bigr\}+\frac{N_h(P_0)}{p_P}-\frac{N_h(P_0)}{p_0}-\frac{\Phi_h(P)}{p_P}(p_0-p_P)\\
&=\frac{1}{p_P}R_{N_h}(P,P_0)+\frac{p_P-p_0}{p_P}\{\Phi_h(P)-\Phi_h(P_0)\}.
\end{align*}
\end{proof}

\subsection{Proof of Lemma \ref{lem:subgroup_remainder_np_tmle}}

\begin{proof}
We begin with the numerator parameter $N_\Psi(P)=E_P[S\,m_P(W)]$. Its canonical gradient at $P$ is
$$
D_{N_\Psi,P}(O)=S\,m_P(W)-N_\Psi(P)+S\left\{\frac{I(A=1)}{g_P(1\mid 1,W)}-\frac{I(A=0)}{g_P(0\mid 1,W)}\right\}\bigl[Y-Q_P(1,W,A)\bigr].
$$
Indeed, since $\Psi(P)=N_\Psi(P)/p_P$, Lemma \ref{lem:ratio_remainder_conditional_mean} implies that
$$
D_{\Psi,P}(O)=\frac{D_{N_\Psi,P}(O)-\Psi(P)\{S-p_P\}}{p_P},
$$
which gives exactly the canonical gradient stated earlier for $\Psi$. We therefore first derive the exact remainder of the numerator:
$$
R_{N_\Psi}(P,P_0)\coloneq N_\Psi(P)-N_\Psi(P_0)+P_0D_{N_\Psi,P}.
$$
Substituting the expression for $D_{N_\Psi,P}$ yields
\begin{align*}
R_{N_\Psi}(P,P_0)&=N_\Psi(P)-N_\Psi(P_0)+P_0\bigl[S\,m_P(W)-N_\Psi(P)\bigr]\\
&+P_0\left[S\left\{\frac{I(A=1)}{g_P(1\mid 1,W)}-\frac{I(A=0)}{g_P(0\mid 1,W)}\right\}\bigl\{Y-Q_P(1,W,A)\bigr\}\right].
\end{align*}
Because $P_0\{S\,m_P(W)-N_\Psi(P)\}=E_0[S\,m_P(W)]-N_\Psi(P)$, the first two terms simplify to $E_0[S\,m_P(W)]-E_0[S\,m_0(W)]$.
Hence
\begin{align*}
R_{N_\Psi}(P,P_0)&=E_0\bigl[S\{m_P(W)-m_0(W)\}\bigr]\\
&+E_0\left[S\left\{\frac{I(A=1)}{g_P(1\mid 1,W)}-\frac{I(A=0)}{g_P(0\mid 1,W)}\right\}\bigl\{Y-Q_P(1,W,A)\bigr\}\right].
\end{align*}
Now condition on $(S,W,A)$. Since the second term is multiplied by $S$, only the subgroup observations contribute, and $E_0\bigl[Y-Q_P(1,W,A)\mid S=1,W,A\bigr]=Q_0(1,W,A)-Q_P(1,W,A)$. Therefore,
\begin{align*}
E_0\left[S\left\{\frac{I(A=1)}{g_P(1\mid 1,W)}-\frac{I(A=0)}{g_P(0\mid 1,W)}\right\}\bigl\{Y-Q_P(1,W,A)\bigr\}\right]&=E_0\left[S\frac{g_0(1\mid 1,W)}{g_P(1\mid 1,W)}\bigl\{Q_0(1,W,1)-Q_P(1,W,1)\bigr\}\right]\\
&-E_0\left[S\frac{g_0(0\mid 1,W)}{g_P(0\mid 1,W)}\bigl\{Q_0(1,W,0)-Q_P(1,W,0)\bigr\}\right].
\end{align*}
Also,
\begin{align*}
E_0\bigl[S\{m_P(W)-m_0(W)\}\bigr]&=E_0\bigl[S\{Q_P(1,W,1)-Q_0(1,W,1)\}\bigr]\\
&-E_0\bigl[S\{Q_P(1,W,0)-Q_0(1,W,0)\}\bigr].
\end{align*}
Combining the two gives
\begin{align*}
R_{N_\Psi}(P,P_0)&=E_0\left[S\left\{1-\frac{g_0(1\mid 1,W)}{g_P(1\mid 1,W)}\right\}\bigl\{Q_P(1,W,1)-Q_0(1,W,1)\bigr\}\right]\\
&-E_0\left[S\left\{1-\frac{g_0(0\mid 1,W)}{g_P(0\mid 1,W)}\right\}\bigl\{Q_P(1,W,0)-Q_0(1,W,0)\bigr\}\right].
\end{align*}
Using
$$
1-\frac{g_0(a\mid 1,W)}{g_P(a\mid 1,W)}=\frac{g_P-g_0}{g_P}(a\mid 1,W),
$$
we obtain
\begin{align*}
R_{N_\Psi}(P,P_0)&=E_0\left[S\frac{g_P-g_0}{g_P}(1\mid 1,W)\bigl\{Q_P(1,W,1)-Q_0(1,W,1)\bigr\}\right]\\
&-E_0\left[S\frac{g_P-g_0}{g_P}(0\mid 1,W)\bigl\{Q_P(1,W,0)-Q_0(1,W,0)\bigr\}\right].
\end{align*}
Finally, apply Lemma \ref{lem:ratio_remainder_conditional_mean} with $h_P(W)=m_P(W)$. Since $\Psi(P)=E_P[m_P(W)\mid S=1]$ and $N_h(P)=N_\Psi(P)$, we obtain
$$
R_\Psi(P,P_0)=\frac{1}{p_P}R_{N_\Psi}(P,P_0)+\frac{p_P-p_0}{p_P}\{\Psi(P)-\Psi(P_0)\}.
$$
Substituting the expression for $R_{N_\Psi}(P,P_0)$ into the above proves the lemma.
\end{proof}

\subsection{Proof of Lemma \ref{lem:subgroup_remainder_tilde_numerator}}

\begin{proof}
From the numerator gradient,
$$
D_{N_A,P}=S\,\tau_{A,\beta_{A,P}}(W)-N_A(P)+E_P[S\phi_A(W)]^\top D_{A,\beta_A,P}.
$$
Consequently,
\begin{align*}
R_{N_A}(P,P_0)=E_0[S\phi_A(W)]^\top(\beta_{A,P}-\beta_{A,0})+E_P[S\phi_A(W)]^\top E_0[D_{A,\beta_A,P}].
\end{align*}
To calculate the last expectation, condition on $W$ and recall that
$$
\bar Q_{w,P}(W,A)=\theta_P(W)+(A-\bar g_P(1\mid W))\tau_{A,\beta_{A,P}}(W).
$$
Then
\begin{align*}
E_0\left[(A-\bar g_P(1\mid W))\{Y-\bar Q_{w,P}(W,A)\}\mid W\right]&=\bar g_0(1\mid W)\{1-\bar g_0(1\mid W)\}\{\tau_{A,P_0}(W)-\tau_{A,\beta_{A,P}}(W)\}\\
&+(\bar g_P-\bar g_0)(1\mid W)(\theta_P-\theta_0)(W)\\
&-(\bar g_P-\bar g_0)^2(1\mid W)
\tau_{A,\beta_{A,P}}(W).
\end{align*}
The projection normal equation at $P_0$ is
$$
E_0\Bigl[\bar g_0(1\mid W)\{1-\bar g_0(1\mid W)\}\phi_A(W)\{\tau_{A,P_0}(W)-\tau_{A,\beta_{A,0}}(W)\}\Bigr]=0.
$$
Therefore, multiplying the conditional identity by $\phi_A(W)$ and integrating yields
\begin{align*}
E_0[D_{A,\beta_A,P}]&=-I_{A,P}^{-1}I_{A,0}(\beta_{A,P}-\beta_{A,0})\\
&+I_{A,P}^{-1}E_0\Bigl[\phi_A(W)\Bigl\{
(\bar g_P-\bar g_0)(1\mid W)(\theta_P-\theta_0)(W)-(\bar g_P-\bar g_0)^2(1\mid W)
\tau_{A,\beta_{A,P}}(W)\Bigr\}\Bigr].
\end{align*}
Substitution gives the claimed expression. The argument uses the projection normal equation and does not require $\tau_{A,P_0}$ to lie in the working model.
\end{proof}

\subsection{Proof of Lemma \ref{lem:subgroup_remainder_bias_numerator}}

\begin{proof}
Using the subgroup-mechanism component from Lemma \ref{lem:can_grad_psi_pound}, the numerator gradient is
\begin{align*}
D_{N_S,P}&=S\,b_P(W)-N_S(P)+p_P D_{\Psi^\#,\mathcal{M}_{S,w},\Pi,P}\\
&+E_P\Bigl[S\bigl\{\Pi_P(0\mid W,0)\phi_S(W,0)-\Pi_P(0\mid W,1)\phi_S(W,1)\bigr\}\Bigr]^\top D_{S,\beta_S,P}.
\end{align*}
Conditioning on $(W,A)$ gives
\begin{align*}
&E_0\left[(S-\Pi_P(1\mid W,A))\{Y-Q_{w,P}(S,W,A)\}\mid W,A\right]\\
&=\Pi_0(1\mid W,A)\{1-\Pi_0(1\mid W,A)\}\{\tau_{S,P_0}(W,A)-\tau_{S,\beta_{S,P}}(W,A)\}\\
&+(\Pi_P-\Pi_0)(1\mid W,A)(\bar Q_P-\bar Q_0)(W,A)\\
&-(\Pi_P-\Pi_0)^2(1\mid W,A)
\tau_{S,\beta_{S,P}}(W,A).
\end{align*}
The projection normal equation at $P_0$ is
\begin{align*}
E_0\Bigl[\Pi_0(1\mid W,A)\{1-\Pi_0(1\mid W,A)\}\phi_S(W,A)\{\tau_{S,P_0}(W,A)-\tau_{S,\beta_{S,0}}(W,A)\}\Bigr]=0.
\end{align*}
It follows that
\begin{align*}
E_0[D_{S,\beta_S,P}]
&=-I_{S,P}^{-1}I_{S,0}(\beta_{S,P}-\beta_{S,0})\\
&+I_{S,P}^{-1}E_0\Bigl[\phi_S(W,A)\Bigl\{
(\Pi_P-\Pi_0)(1\mid W,A)(\bar Q_P-\bar Q_0)(W,A)\\
&-(\Pi_P-\Pi_0)^2(1\mid W,A)
\tau_{S,\beta_{S,P}}(W,A)\Bigr\}\Bigr].
\end{align*}

Next, expanding $b_P-b_0$ gives
\begin{align*}
E_0[S(b_P-b_0)]&=E_0\Bigl[S\bigl\{\Pi_0(0\mid W,0)\phi_S(W,0)-\Pi_0(0\mid W,1)\phi_S(W,1)\bigr\}\Bigr]^\top(\beta_{S,P}-\beta_{S,0})\\
&+E_0\Bigl[\bar\Pi_0(1\mid W)\sum_{a=0}^1(2a-1)(\Pi_P-\Pi_0)(1\mid W,a)\tau_{S,\beta_{S,P}}(W,a)\Bigr].
\end{align*}
Conditioning on $W$ in the subgroup-mechanism gradient term yields
\begin{align*}
p_P E_0[D_{\Psi^\#,\mathcal{M}_{S,w},\Pi,P}]&=-E_0\Biggl[\bar\Pi_P(1\mid W)\sum_{a=0}^1(2a-1)\frac{\bar g_0(a\mid W)}{\bar g_P(a\mid W)}(\Pi_P-\Pi_0)(1\mid W,a)\tau_{S,\beta_{S,P}}(W,a)\Biggr].
\end{align*}
Finally, the numerator remainder satisfies
\begin{align*}
R_{N_S}(P,P_0)
&=E_0[S(b_P-b_0)]+p_P E_0[D_{\Psi^\#,\mathcal{M}_{S,w},\Pi,P}]\\
&+E_P\Bigl[S\bigl\{\Pi_P(0\mid W,0)\phi_S(W,0)-\Pi_P(0\mid W,1)\phi_S(W,1)\bigr\}\Bigr]^\top E_0[D_{S,\beta_S,P}].
\end{align*}
Substitution of the preceding identities proves the lemma.
\end{proof}

\subsection{Derivation of Remark \ref{rem:balanced_bias_remainder}}

Under the stated restrictions, let $v$ satisfy
$\phi_S(W,a)^\top v=2-4a$.
Balanced randomization implies
\begin{align*}
I_{S,P}v
&=E_P\Bigl[
\bar\Pi_P(1\mid W)\bar\Pi_P(0\mid W)
\{\phi_S(W,0)-\phi_S(W,1)\}
\Bigr]\\
&=E_P\Bigl[S\bigl\{
\Pi_P(0\mid W,0)\phi_S(W,0)
-\Pi_P(0\mid W,1)\phi_S(W,1)
\bigr\}\Bigr].
\end{align*}
The same identity holds at $P_0$, so the term multiplying
$\beta_{S,P}-\beta_{S,0}$ in
Lemma \ref{lem:subgroup_remainder_bias_numerator} vanishes.
The row vector premultiplying the next expectation equals $v^\top$.
Conditioning that term on $W$ therefore gives
\begin{align*}
&E_0\Bigl[(\bar\Pi_P-\bar\Pi_0)(1\mid W)
\bigl\{(\bar Q_P-\bar Q_0)(W,0)
-(\bar Q_P-\bar Q_0)(W,1)\bigr\}\Bigr]\\
&\quad-E_0\Bigl[(\bar\Pi_P-\bar\Pi_0)^2(1\mid W)
\bigl\{\tau_{S,\beta_{S,P}}(W,0)
-\tau_{S,\beta_{S,P}}(W,1)\bigr\}\Bigr].
\end{align*}
The final expectation in the lemma becomes
\[
E_0\Bigl[(\bar\Pi_P-\bar\Pi_0)^2(1\mid W)
\bigl\{\tau_{S,\beta_{S,P}}(W,0)
-\tau_{S,\beta_{S,P}}(W,1)\bigr\}\Bigr].
\]
The squared-membership-error terms cancel, proving the simplification.

\section{TMLE-style targeting of the working model coefficients}
In \cite{vanderlaan2024atmle}, the authors consider refitting the working-model coefficients so that the empirical mean of the corresponding components of the EIC is solved to zero. This is done by refitting the coefficients of the working models using the unpenalized version of the loss functions. In this appendix, we propose an alternative, TMLE-style approach for targeting these coefficients.

\subsection{Targeting working model coefficients for the pooled-ATE parameter $\tilde{\Psi}$}\label{app:tmle_psi_tilde}

Recall that the $\beta$-component of the EIC is given by
$$
D_{\mathcal{M}_{A,w},\beta,P}=\tilde{D}_{\beta,P}^\top E_P\phi(W),
$$
where
$$
\tilde{D}_{\beta,P}=I_{A,P}^{-1}(A-g_P(1\mid W))\phi(W)\left[Y-\theta_P(W)-(A-g_P(1\mid W))\sum_j\beta_P(j)\phi_j(W)\right],
$$
and $I_{A,P}=E_Pg_P(1-g_P)(1\mid W)\phi\phi^\top(W)$. Define $C_n=I_{A,n}^{-1}E_n\phi(W).$ Consider a one-dimensional submodel
$\beta_{n,\epsilon}=\beta_n+\epsilon C_n$ for $\beta$ through its initial $\beta_n$ in the direction $C_n$. Then we have that
$$
\tau_{A,\beta_{n,\epsilon}}(W)=\sum_j\beta_{n,\epsilon}(j)\phi_j(W)=\tau_{A,\beta_n}(W)+\epsilon\phi(W)C_n.
$$
We define
$$
H_n=(A-g_n(1\mid W))\phi(W)C_n,
\qquad
R_n=Y-\theta_n(W)-(A-g_n(1\mid W))\tau_{A,\beta_n}(W).
$$
Then, we can represent the $\beta$-component of the EIC as $D_{\mathcal{M}_{A,w},\beta,P}=H_PR_P.$ Note that 
\begin{align*}
R_{n,\epsilon}&=Y-\theta_n(W)-(A-g_n(1\mid W))\tau_{A,\beta_{n,\epsilon}}(W)\\
&=R_n-\epsilon(A-g_n(1\mid W))\phi(W)C_n\\
&=R_n-\epsilon H_n.
\end{align*}
Therefore, the closed-form solution $\epsilon^\star$ of $\epsilon$ that solves the estimating equation $0=P_nH_n(R_n-\epsilon H_n)$ is given by
$$
\epsilon^{\star}=\frac{P_n(H_nR_n)}{P_n(H_n^2)}.
$$
For $\tilde{\Psi}_2$, one could replace the definition of $C_n$ by 
$$
I_{A,n}^{-1}E_n\left[\frac{S}{P_n(S=1)}\phi_A(W)\right]^\top.
$$

\subsection{Targeting working model coefficients for the bias parameter $\Psi^\#$}\label{app:tmle_psi_pound}

Recall that the $\beta$-component of the EIC is given by
$$
D^\#_{\mathcal{M}_{S,w},\beta,P}=D_{\beta,P}^\top\left[E_P\Pi_P(0\mid W,0)\phi(W,0)-E_P\Pi_P(0\mid W,1)\phi(W,1)\right],
$$
where
$$
D_{\beta,P}=I_{S,P}^{-1}(S-\Pi_P(1\mid W,A))\phi(W,A)\left[Y-\bar Q_P(W,A)-(S-\Pi_P(1\mid W,A))\sum_j\beta_P(j)\phi_j(W,A)\right],
$$
and $I_{S,P}=E_P\Pi_P(1-\Pi_P)(1\mid W,A)\phi\phi^\top(W,A)$. Define
$$
C_n=I_{S,n}^{-1}\left[E_n\Pi_n(0\mid W,0)\phi(W,0)-E_n\Pi_n(0\mid W,1)\phi(W,1)\right].
$$
Consider a one-dimensional submodel $\beta_{n,\epsilon}=\beta_{n}+\epsilon C_n$ for $\beta$ in the direction $C_n$. Then we have that
$$
\tau_{S,\beta_{n,\epsilon}}(W,A)=\sum_j\beta_{n,\epsilon}(j)\phi_j(W,A)=\tau_{S,\beta_n}(W,A)+\epsilon\phi(W,A)C_n.
$$
Define
\begin{align*}
H_n&=(S-\Pi_n(1\mid W,A))\phi(W,A)C_n\\
R_n&=Y-\bar{Q}_n(W,A)-(S-\Pi_n(1\mid W,A))\tau_{S,\beta_n}(W,A),
\end{align*}
we can then write the EIC as $D^\#_{\mathcal{M}_{S,w},\beta,P}=H_PR_P$. Note that 
\begin{align*}
R_{n,\epsilon}&=Y-\bar{Q}_n(W,A)-(S-\Pi_n(1\mid W,A))\tau_{S,\beta_{n,\epsilon}}(W,A)\\
&=R_n-\epsilon(S-\Pi_n(1\mid W,A))\phi(W,A)C_n\\
&=R_n-\epsilon H_n.
\end{align*}
Therefore, the closed-form solution $\epsilon^\star$ of $\epsilon$ that solves the estimating equation $0=P_nH_n(R_n-\epsilon H_n)$ is again given by
$$
\epsilon^{\star}=\frac{P_n(H_nR_n)}{P_n(H_n^2)}.
$$
\section{Implementation details}\label{app:implementation_details}

\subsection{Estimator configurations for simulations}

For all candidate estimators considered in our simulation studies, nuisance functions are estimated using a discrete Super Learner (\texttt{R} packages \texttt{SuperLearner} \citep{polley_package_2019} and \texttt{sl3} \citep{coyle_sl3_2021}) with a 3-fold cross-validation scheme. The Super Learner library includes generalized linear models \citep{nelder_generalized_1972}, multivariate adaptive regression splines \citep{friedman_multivariate_1991}, and generalized additive models \citep{hastie_generalized_1986,hastie_gam_2017}. The final learner is selected based on cross-validated log-likelihood loss. All candidate learners are implemented with their default hyperparameters. For A-TMLE, the working models for both the conditional average treatment effect $\tau_{A}(W)=E(Y\mid A=1,W)-E(Y\mid A=0,W)$ and the conditional average subgroup effect $\tau_S(A,W)=E(Y\mid S=1,W,A)-E(Y\mid S=0,W,A)$ were estimated using HAL (\texttt{R} package \texttt{hal9001} \citep{coyle_hal9001_2025,hejazi_hal9001_2020}) with the group Lasso (\texttt{R} package \texttt{sparsegl} \citep{xiaoxuan_sparsegl_2024}) screening procedure described in Subsection \ref{sec:group_lasso}. Specifically, we first perform screening using a first-order basis with \texttt{max\_degree=5}, \texttt{smoothness\_orders=1}, and \texttt{num\_knots=5}, followed by HAL fitting on the selected subspace with \texttt{max\_degree=3}, \texttt{smoothness\_orders=1}, and \texttt{num\_knots=20}. The targeting step is implemented using the TMLE-style procedure described in Subsection \ref{subsec:tmle_style_targeting}. To protect against inverse weighting of small estimated probabilities (treatment mechanism or subgroup membership mechanism, whichever is relevant), we use the fixed truncation rule $5/(\sqrt{n}\log n)$, which was empirically shown to be a robust practical default in many settings \citep{xu_investigating_2026}.

\subsection{Estimator configurations for the case study}

For all candidate estimators considered in our LEADER case study, nuisance functions are estimated using a discrete Super Learner with a 3-fold cross-validation scheme. The Super Learner library includes generalized linear models \citep{nelder_generalized_1972}, Lasso regression \citep{tibshirani_regression_1996,friedman_regularization_2010}, Bayesian additive regression trees \citep{chipman_bart_2010,kapelner_bartmachine_2016}, multivariate adaptive regression splines \citep{friedman_multivariate_1991}, generalized additive models \citep{hastie_generalized_1986,hastie_gam_2017}, and random forests \citep{breiman_random_2001,marvin_ranger_2017}, with the final learner selected by cross-validated log-likelihood loss. A-TMLE working models are fitted the same way as in simulations. The same truncation rule $5/(\sqrt{n}\log n)$ is also applied.

\end{document}